%% file: main.tex
\documentclass[acmsmall,screen,nonacm]{acmart}

\acmConference[PLDI'26]{PLDI'26}{June 2026}{Boulder, CO}
\acmYear{2026}
\acmISBN{} 
\acmDOI{} 
\startPage{1}

\setcopyright{none}

\usepackage{booktabs}   
\usepackage{subcaption} 

\usepackage{semantic}

\input{setup.tex}

\usepackage{booktabs}

\begin{document}

\title[Practical Formal Verification for MLIR Programs]{Practical Formal Verification for MLIR Programs}

\author{Emily Tucker}
\affiliation{\institution{Georgia Institute of Technology}\country{USA}}
\email{emily.tucker@gatech.edu}

\author{Louis-No{\"e}l Pouchet}
\affiliation{\institution{Colorado State University}\country{USA}}
\email{pouchet@colostate.edu}

\author{Erika Hunhoff}
\affiliation{\institution{AMD}\country{USA}}
\email{erika.hunhoff@amd.com}

\author{Stephen Neuendorffer}
\affiliation{\institution{AMD}\country{USA}}
\email{stephen.neuendorffer@amd.com}

\author{Erwei Wang}
\affiliation{\institution{AMD}\country{USA}}
\email{erwei.wang@amd.com}

\begin{abstract}
\input{abstract.tex}
\end{abstract}



\maketitle

\section{Introduction}
\label{sec:introduction}

\input{introduction.tex}

\section{Motivation and Overview}
\label{sec:motivation}

\input{motivation.tex}

\section{Formalization of Program Correctness Verification}
\label{sec:formal}
\input{peqc-formal.tex}

\section{\PEQC Implementation}
\label{sec:implementation}
\input{peqc-implementation.tex}

\section{Enabling Correctness Checking of MLIR Programs}
\label{sec:mlirsupport}

\input{mlirsupport.tex}

\section{Experiments}
\label{sec:experiments}

\input{experiments.tex}

\section{Related Work}
\label{sec:related}
\input{related.tex}

\section{Conclusion}
\label{sec:conclusion}

\input{conclusion.tex}

\balance
\bibliography{bibs/sc24.bib,bibs/lnp.bib,bibs/gabriel.bib,bibs/ics15.bib,bibs/gabrielsc24.bib,bibs/ierefs.bib,bibs/fpga23.bib,bibs/refs.bib,bibs/rewriting.bib,bibs/emily.bib,bibs/vf-biblio.bib}

\newpage
\appendix

\input{appendix.tex}
\end{document}

%% file: setup.tex
\usepackage{caption}
\usepackage{amsmath,amsfonts}
\usepackage{algorithmic}
\usepackage{graphicx}
\usepackage{textcomp}
\usepackage{hyperref}
\usepackage{listings}
\usepackage{placeins}
\usepackage{url}
\usepackage{csvsimple}
\usepackage{adjustbox}
\usepackage{float}
\usepackage{array}
\usepackage{titlesec}
\usepackage{units}
\newcolumntype{P}[1]{>{\centering\arraybackslash}p{#1}}

\usepackage{paralist}
\usepackage{balance}

\usepackage{enumitem}      
\setitemize{noitemsep,topsep=0pt,parsep=0pt,partopsep=0pt}

\usepackage{backnaur}
\usepackage{tikz}
\usetikzlibrary{shapes.geometric,fit}
\usepackage{subfiles}
\usepackage{soul}
\usepackage{lscape}
\usepackage[ruled,lined,boxed,linesnumbered]{algorithm2e}
\usepackage{svg}
\usepackage{adjustbox}
\usepackage{float}
\usepackage{placeins}
\usetikzlibrary{arrows}
\usepackage{bookmark}
\usepackage{epigraph}

\definecolor{crimson}{rgb}{0.86, 0.08, 0.24}
\definecolor{crimsonglory}{rgb}{0.75, 0.0, 0.2}
\definecolor{light}{gray}{0.92}
\definecolor{dark}{gray}{0.50}
\definecolor{red}{rgb}{1,0,0}
\definecolor{azure}{rgb}{0, 0.5,1.0}
\definecolor{headertable}{gray}{0.65}
\definecolor{mygray}{gray}{1}

\lstdefinestyle{GACstyle}{
    basicstyle=\footnotesize\ttfamily,
    breakatwhitespace=false,
    breaklines=false,
    captionpos=b,
    keepspaces=true,
    numbers=left,
    numbersep=5pt,
    showspaces=false,
    showstringspaces=false,
    showtabs=false,
    tabsize=2,
    xleftmargin=0.035\textwidth,
    language=python,
    linewidth=1.0\textwidth,
    aboveskip=0.5cm,
}

\lstdefinestyle{MinStyle}{
    basicstyle=\footnotesize\ttfamily,
    breakatwhitespace=false,
    breaklines=false,
    captionpos=b,
    keepspaces=true,
    numbers=left,
    numbersep=5pt,
    showspaces=false,
    showstringspaces=false,
    showtabs=false,
    tabsize=2,
    xleftmargin=0.035\textwidth,
    language=C,
    linewidth=1.0\textwidth,
    aboveskip=0.5cm,
}

\lstdefinestyle{customc}{
  belowcaptionskip=1\baselineskip,
  numbers=left,
  numbersep=8pt,
  stepnumber=1,
  float=tb,
  tabsize=4,
  breaklines=true,
  frame=single,
  xleftmargin=\parindent,
  language=C,
  showstringspaces=false,
  basicstyle=\footnotesize\ttfamily,
  keywordstyle=\bfseries\color{green!40!black},
  commentstyle=\itshape\color{purple!40!black},
  identifierstyle=\color{blue},
  stringstyle=\color{orange},
  morekeywords={\_\_m256i,\_\_m256,\_\_m128,\_\_m128i,\_\_m512,\_\_m512i,defined},
  columns = fullflexible,
}

\lstdefinestyle{customcminipage}{
  belowcaptionskip=1\baselineskip,
  numbers=left,
  numbersep=8pt,
  stepnumber=1,
  tabsize=4,
  breaklines=true,
  frame=single,
  xleftmargin=\parindent,
  language=C,
  showstringspaces=false,
  basicstyle=\footnotesize\ttfamily,
  keywordstyle=\bfseries\color{green!40!black},
  commentstyle=\itshape\color{purple!40!black},
  identifierstyle=\color{blue},
  stringstyle=\color{orange},
  morekeywords={\_\_m256i,\_\_m256,\_\_m128,\_\_m128i,\_\_m512,\_\_m512i,defined},
  columns = fullflexible,
}

\lstdefinestyle{custommrt}{
  belowcaptionskip=1\baselineskip,
  float=tb,
  numbers=left,
  numbersep=8pt,
  stepnumber=1,
  tabsize=4,
  breaklines=true,
  frame=single,
  xleftmargin=\parindent,
  language=C,
  columns = fullflexible,
}

\lstdefinestyle{customasm}{
  belowcaptionskip=1\baselineskip,
  float=tb,
  numbers=left,
  numbersep=8pt,
  stepnumber=1,
  tabsize=4,
  frame=single,
  xleftmargin=\parindent,
  language=[x86masm]Assembler,
  basicstyle=\footnotesize\ttfamily,
  commentstyle=\itshape\color{purple!40!black},
  identifierstyle=\color{blue},
  stringstyle=\color{orange},
  keywordstyle=\bfseries\color{red!40!black},
  keywordstyle=[1]\bfseries\color{red!40!black},
  morekeywords=[1]{QWORD,PTR,DWORD,YMMWORD,XMMWORD,ZMMWORD},
  deletekeywords=[1]{r1,r2,r3,r4,r5,r6,r7,r8,r9,r10,r11,r12,r13,r14,r15,
  rax,rcx,rdx,rbx,rsp,rbp,rsi,rdi,rip,eax,ecx,edx,ebx,esp,ebp,esi,edi},
  keywordstyle=[2]\bfseries\color{blue!40!black},
  morekeywords=[2]{
  r1,r2,r3,r4,r5,r6,r7,r8,r9,r10,r11,r12,r13,r14,r15,
  rax,rcx,rdx,rbx,rsp,rbp,rsi,rdi,rip,eax,ecx,edx,ebx,esp,ebp,esi,edi,
  xmm0,xmm1,xmm2,xmm3,xmm4,xmm5,xmm6,xmm7,xmm8,xmm9,xmm10,xmm11,xmm12,xmm13,xmm14,xmm15,
  ymm0,ymm1,ymm2,ymm3,ymm4,ymm5,ymm6,ymm7},
  alsoletter=[2]{.,0,1,2,3,4,5,6,7,8,9},
  morecomment=[n][\bfseries\color{red!40!black}]{[}{]},
  morecomment=[f][\itshape\color{purple!40!black}][0]{;;},
  literate=*{,}{,}{1},
  columns = fullflexible,
}

\lstdefinestyle{custompython}{
  belowcaptionskip=1\baselineskip,
  float=tb,
  numbers=left,
  numbersep=8pt,
  stepnumber=1,
  tabsize=4,
  frame=single,
  xleftmargin=\parindent,
  language=Python,
  basicstyle=\footnotesize\ttfamily,
  commentstyle=\fontseries{lc}\selectfont\itshape\color{purple!40!black},
  identifierstyle=\color{blue},
  stringstyle=\color{orange},
}

\newcommand{\FIXME}[1]{\par\textbf{FIXME:~ #1}\par}

\newtheorem{theorem}{Theorem}[section]

\newtheorem{lemma}[theorem]{Lemma}
\newtheorem{corollary}[theorem]{Corollary}

\newtheorem{definition}[theorem]{Definition}

\usepackage{wrapfig}
\usepackage{xspace}
\usepackage{semantic}

\newcommand{\PEQC}{\textsf{PEQC-MLIR}\xspace}

\newcommand{\MLIRDIAL}[1]{\texttt{#1}}

%% file: abstract.tex
Optimizing compilers have become a cornerstone for high-performance program generation in research and industry. Optimizations, including those implemented manually by a user and those target-specific and non-target-specific, are used to transform programs to achieve good performance. Although these optimizations are necessary for performance, assessing their correctness has remained a major challenge; the risk of incorrect code being deployed increases with unproven optimization flows.

In this work, we target the formal verification of correctness of a transformed program by computing whether a pair of programs are semantically equivalent, one being a transformed version of the other. We restrict the class of programs supported to enable a hybrid concrete-symbolic interpretation approach to equivalence, which in turn is mostly agnostic to how the programs are implemented (syntax, schedule, storage, etc.). This approach can show equivalence in linear time with respect to the operations executed by the programs. We develop a verifier for a meaningful subset of MLIR, and report on the verification of the AMD MLIR-AIR and MLIR-AIE toolchains, as well as the standard mlir-opt on hundreds of benchmarks variants.

%% file: introduction.tex
It has long been a requirement to optimize a program for performance by modifying its implementation: altering the order of operations (e.g., loop tiling), parallelizing operations and implementing their synchronizations, and exposing instruction-level parallelism (e.g., via unrolling) are typical mandatory optimizations for performance of numerical applications. However, the issue of determining whether the transformed program maintains the semantics of the original program remains a fundamental challenge. No global system can systematically prove the correctness of arbitrary transformations applied to arbitrary programs. 

To mitigate the risks of incorrect transformations being deployed, numerous approaches have been developed, typically targeting a specific framework or program class. The techniques range from simple testing to more complex fuzzing \cite{chen2013taming,marcozzi2019compiler}, where the outputs produced by both programs are compared, facilitating the detection of bugs. Compiler optimizations can be validated using translation validation techniques, such as Alive2 \cite{lopes2021alive2} for LLVM IR. A pair of programs can be checked for symbolic equivalence using symbolic execution and employing SMT queries to reason on program properties, as with e.g., UC-KLEE \cite{ramos2015under,klee-web} or CIVL \cite{siegel2015civl}. 
Techniques that prove equivalence between a pair of affine programs, such as ISA for a class of affine program transformations \cite{verdoolaege2012equivalence}, have been proposed but are typically very limited in scope.
These techniques all present different trade-offs in terms of coverage (e.g., types of input programs, types of transformations performed), speed, and accuracy. In this work, we target the verification of correctness of transformations in a manner mostly agnostic to how they were produced: taking as input a pair of programs, one being a transformed version of the other, we aim to prove their semantic equivalence for a practical class of programs that includes highly optimized implementations of key AI workloads.

PEQC \cite{fpga24} is a hybrid concrete-symbolic interpreter which concretely evaluates (and simplifies) the input program and builds a symbolic representation of all other operations that contribute to the production of the final outputs \cite{fpga24}. These representations are then checked for equivalence, either proving correctness or exposing mismatches. However, PEQC trades generality in the program structure for restrictions of the class of programs handled: it typically requires fixed problem sizes (e.g., in a matrix-multiply program, the sizes of matrices are fixed and known at compile-time, and correctness guarantees apply to any input matrices using this specific size).

In this work, we leverage the approach developed in PEQC to develop a novel framework for verification via hybrid concrete-symbolic interpretation. We significantly revisit and extend PEQC in \PEQC, the system presented in this paper. \PEQC supports verification of: (1) general, hierarchical parallel tasks, (2) arbitrary synchronization with binary (\texttt{wait}/\texttt{set}) and counting (\texttt{acquire}/\texttt{release}) semaphores, and (3) a class of symbolic conditional control-flow, significantly broadening the scope of applicability of this approach. Furthermore, we develop a complete formalization of a proxy parallel language representing \PEQC's input, detailing the process of ensuring correctness of parallel programs and computing their equivalence, and proving correctness of our approach.  

To enable wide applicability of this approach, we develop a verifier for a set of MLIR \cite{lattner2021mlir} programs, including concurrent programs using the \MLIRDIAL{async} dialect \cite{mlir-web}, and the ability to verify their equivalence to a sequential version.
We use this verifier to expose several important bugs in tools such as \texttt{mlir-opt}. 
We also contribute support for the MLIR-AIR \cite{mlir-air,wang2025loop} and MLIR-AIE \cite{mlir-aie-github} systems to map computations to AMD AI Engines \cite{aiearch.micro.2024}, enabling the verification of correctness for practical workloads deployed on the AI Engines. We make the following contributions:
\begin{itemize}[leftmargin=*]
\item We present an end-to-end, fully automated approach to verify the equivalence between a pair of (concurrent) programs, under practical restrictions.
\item We formalize our hybrid concrete-symbolic approach and prove that it correctly detects and handles deterministic parallel programs and their equivalence in Sec.~\ref{sec:formal}.
\item We introduce our approach to the correctness verification of a rich class of MLIR programs and their transformations, by successive lowering and conversion of user-defined dialects to standard dialects, and eventually using Parallel Intermediate Representation (PIR) of \PEQC in Sec.~\ref{sec:mlirsupport}.
\item We present extensive evaluation of \PEQC in Sec.~\ref{sec:experiments}, over hundreds of MLIR program variants generated using \texttt{mlir-opt} and Polygeist, uncovering a variety of bugs and lack of analysis to adequately apply transformations. We present our verification for the AMD MLIR-AIR and MLIR-AIE toolchains for mapping AI workloads to the AI Engine, exposing subtle concurrency bugs that were then fixed in these toolchains.
\end{itemize}

%% file: motivation.tex
This section presents the target problem, associated challenges, and outlines our proposed solution.
Our objective is to prove the correctness of programs by enabling the seamless integration of correctness verification within a progressive optimization flow. 
Starting from a high-level algorithmic description, high-performance programs are typically progressively optimized and refined by automatic or manual modifications to the program. 

The Multi-Level IR (MLIR) ecosystem exemplifies this flow. 
MLIR \cite{lattner2021mlir} is a compiler infrastructure designed to represent programs using multiple intermediate representations at different levels of abstraction within a single framework. 
Rather than relying on a monolithic IR, MLIR organizes program representations into dialects, each tailored to express specific program semantics, such as tensor algebra, loop nests, or structured control flow \cite{mlir-web,llvm-web}.
A central concept in MLIR is progressive lowering, where programs are gradually transformed from high-level, semantic-rich representations into lower-level forms that make execution details explicit. Fig.~\ref{fig:mlirexample} provides an illustration of this.
MLIR allows analyses and transformations to be applied at the abstraction level where they are most effective, while preserving program structure as long as possible during compilation.

\begin{figure}[h!tb]
\begin{minipage}[c]{0.99\textwidth}
\begin{minipage}[b]{0.53\textwidth}
\begin{minipage}[b]{0.99\textwidth}
\hspace{-.5cm}\begin{lstlisting}[basicstyle=\tiny,numbers=none]
@linalg_structured_op
def matmul_mono(
    A=TensorDef(T, S.M, S.K),
    B=TensorDef(T, S.K, S.N),
    C=TensorDef(T, S.M, S.N, output=True),
):
    domain(D.m, D.n, D.k)
    C[D.m, D.n] += A[D.m, D.k] * B[D.k, D.n]

@module_builder
def matmul_on_tensors(m, n, k):
    dtype = IntegerType.get_signless(width=32)

    @func.FuncOp.from_py_func(
        MemRefType.get((m, k), dtype), 
        MemRefType.get((k, n), dtype)
    )
    def forward(lhs, rhs):
        out = memref.AllocOp(MemRefType.get((m, n),
        dtype), [], [])
        zero = arith.ConstantOp(dtype, 0)
        zero_fill = linalg.fill(zero, outs=[out])
        matmul_mono(lhs, rhs, outs=[out])
        return out
\end{lstlisting}        
\end{minipage}
\hrule
\begin{minipage}[b]{0.99\textwidth}
\begin{lstlisting}[basicstyle=\tiny,numbers=none]
#map = affine_map<(d0, d1, d2) -> (d0, d2)>
#map1 = affine_map<(d0, d1, d2) -> (d2, d1)>
#map2 = affine_map<(d0, d1, d2) -> (d0, d1)>
module {
 func.func @forward(%arg0: memref<16x16xi32>, 
 %arg1: memref<16x16xi32>) -> memref<16x16xi32> {
   %alloc = memref.alloc() : memref<16x16xi32>
   %c0_i32 = arith.constant 0 : i32
   linalg.fill ins(%c0_i32 : i32) outs(%alloc : 
   memref<16x16xi32>)
   linalg.generic {indexing_maps = [#map, #map1, #map2], 
   iterator_types = ["parallel", "parallel",
   "reduction"]} ins(%arg0, %arg1 : memref<16x16xi32>,
   memref<16x16xi32>) outs(%alloc : memref<16x16xi32>) 
   {
   ^bb0(%in: i32, %in_0: i32, %out: i32):
     %0 = arith.muli %in, %in_0 : i32
     %1 = arith.addi %out, %0 : i32
     linalg.yield %1 : i32
   }
   return %alloc : memref<16x16xi32>
 }
}
\end{lstlisting}
\end{minipage}
\end{minipage}
\vline
\begin{minipage}[b]{0.49\textwidth}
\begin{lstlisting}[basicstyle=\tiny,numbers=none]
module {
 func.func @forward(%arg0: memref<16x16xi32>,
 %arg1: memref<16x16xi32>) -> memref<16x16xi32> {
   %c0_i32 = arith.constant 0 : i32
   %alloc = memref.alloc() : memref<16x16xi32>
   %c0 = arith.constant 0 : index
   %c16 = arith.constant 16 : index
   %c1 = arith.constant 1 : index
   scf.for %arg2 = %c0 to %c16 step %c1 {
     %c0_3 = arith.constant 0 : index
     %c16_4 = arith.constant 16 : index
     %c1_5 = arith.constant 1 : index
     scf.for %arg3 = %c0_3 to %c16_4 step %c1_5 {
       memref.store %c0_i32, %alloc[%arg2, %arg3] : 
       memref<16x16xi32>
     }
   }
   %c0_0 = arith.constant 0 : index
   %c16_1 = arith.constant 16 : index
   %c1_2 = arith.constant 1 : index
   scf.for %arg2 = %c0_0 to %c16_1 step %c1_2 {
     %c0_3 = arith.constant 0 : index
     %c16_4 = arith.constant 16 : index
     %c1_5 = arith.constant 1 : index
     scf.for %arg3 = %c0_3 to %c16_4 step %c1_5 {
       %c0_6 = arith.constant 0 : index
       %c16_7 = arith.constant 16 : index
       %c1_8 = arith.constant 1 : index
       scf.for %arg4 = %c0_6 to %c16_7 step %c1_8 {
         %0 = memref.load %arg0[%arg2, %arg4] : 
         memref<16x16xi32>
         %1 = memref.load %arg1[%arg4, %arg3] : 
         memref<16x16xi32>
         %2 = memref.load %alloc[%arg2, %arg3] : 
         memref<16x16xi32>
         %3 = arith.muli %0, %1 : i32
         %4 = arith.addi %2, %3 : i32
         memref.store %4, %alloc[%arg2, %arg3] :
         memref<16x16xi32>
       }
     }
   }
   return %alloc : memref<16x16xi32>
 }
}  
\end{lstlisting}
\end{minipage}
\end{minipage}
\caption{\label{fig:mlirexample}MLIR example. Top left: the Python program to describe matrix-multiplication, which is then expressed as a MLIR Linalg program. 
Bottom left: the resulting Linalg program, for a problem size of 16x16. 
Right: the scf program obtained by lowering Linalg to affine with --linalg-to-affine-loops and then to scf with --lower-affine. \PEQC targets the verification of correctness for MLIR programs by exploiting such lowerings, as described in Sec.~\ref{sec:mlirsupport}.}
\end{figure}

AMD has developed the MLIR-AIR \cite{mlir-air-github,wang2025loop} and MLIR-AIE \cite{mlir-aie-github, hunhoff2025} compilation flows, for production use to map AI workloads to the AMD AI Engine \cite{wang2025loop} which is integrated in AMD Ryzen AI processors. Starting from a high-level description of the computation, for example \texttt{matmul} in the MLIR \MLIRDIAL{linalg} dialect or its equivalent Python description, the program is successively lowered, transformed, and optimized into an AIE program which maps computation and data movement to concurrent hardware units, synchronizing them using FIFOs and counting semaphores (\texttt{acquire}/\texttt{release}). A 2-line sequential \texttt{matmul} program operating on tensors turns into a multi-buffered, highly concurrent program of nearly 1k lines. These transformations take place at the MLIR level and may be either fully automated by a compiler, implemented by AI agents, or implemented manually by human engineers.
At every transformation stage, we aim to provide a system that assesses the correctness of a transformation, be it manual, automated or produced by an open-source or black-box tool.
We validate correctness by computing whether the original and transformed programs are semantically equivalent: they must produce the same results for all output variables.

\paragraph*{Challenges} Numerous challenges must be addressed to achieve these objectives. First, \emph{our system should be robust to any means used to transform a program}, to reflect typical design flows. This includes programs produced by generative AI. We specifically do not rely on the user to define high-level semantic properties that an implementation should match or how a transformation should modify the program. Instead, we directly prove equivalence for a class of programs without user input. 
Second, \emph{our system should be robust to a wide class of optimizations, which includes arbitrary data layout changes, bufferization schemes, data reuse implementation, etc.}. These transformations are critical for high-performance execution of AI workloads on accelerators, such as the AI Engine, where the orchestration of data movement requires careful fine-tuning. Such a system should be future-proof to new variations of these optimizations, and should avoid the need to specify the transformation effects to enable verification (e.g., \cite{bhat24,yin25}). Third, \emph{our system should be fully capable of reasoning about the correctness of concurrency, parallelization, and arbitrary synchronization approaches}, including assessing the correctness of parallelizing a sequential program. This support is fundamental because most practical hardware is concurrent, and often has numerous hardware replications of computation/communication units that must be exploited for high performance.
Fourth, \emph{our system should provide stronger guarantees than testing alone}. Although testing programs on pre-defined inputs/outputs can expose bugs (e.g., concurrency bugs with dynamic race detectors \cite{gu18}), this is not a guarantee of correctness. The ability to reason \emph{symbolically} on possible inputs enables providing correctness guarantees for sets of possible concrete values these symbols may take, and not simply exposing bugs.

\begin{figure}[h!tb]
\includegraphics[width=12cm]{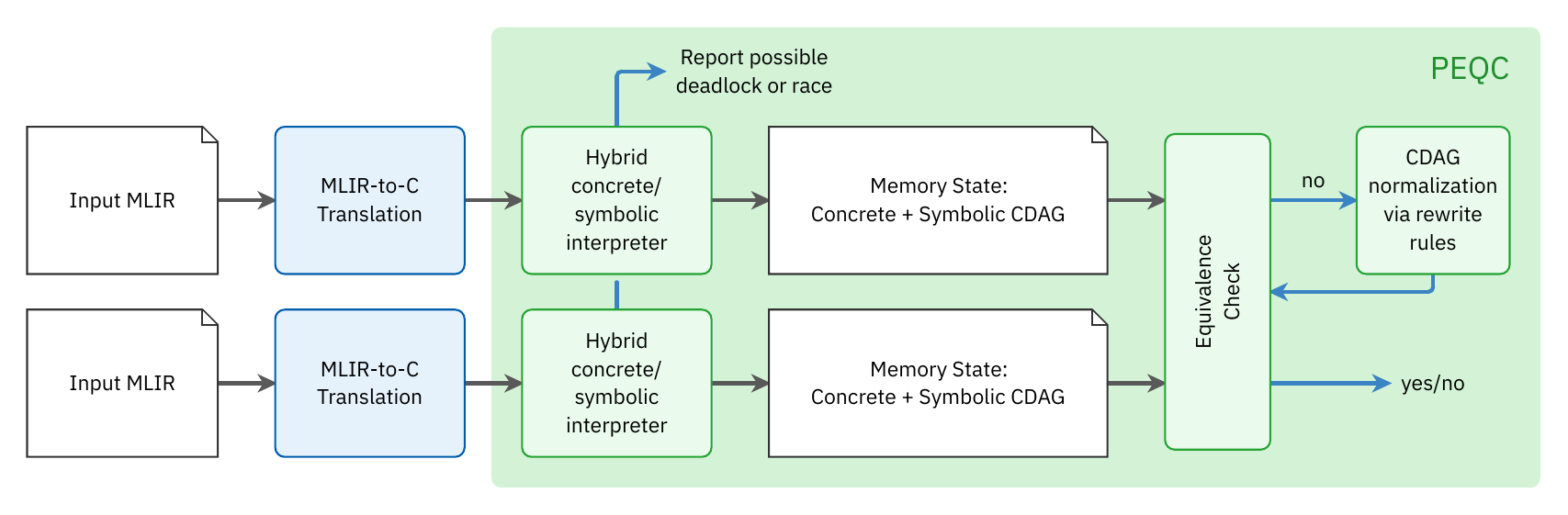}
\vspace{-.5cm}
\caption{\label{fig:flow1} Overall flow for verification with \PEQC, and the MLIR-to-PIR translators presented in Sec.~\ref{sec:mlirsupport}}
\end{figure}
\paragraph*{\PEQC} In this paper, we present \PEQC, which provides a solution to the challenges above by carefully restricting the class of programs supported while still enabling full support of representative AI workloads. 
\PEQC is based on prior work by Pouchet, Tucker et al. \cite{fpga24,tucker2025verification}, which we significantly extend to support a much wider class of (concurrent) programs. Their system implements a hybrid concrete-symbolic interpretation of the program, simplifying out any expressions or control-flow that can be computed at compile time for the program (via partial concrete interpretation), and computing symbolic values/expressions for everything else. 
Through this hybrid interpretation stage, a symbolic representation of the program is built. Equivalence can then be shown by proving these representations are semantically equivalent (which is trivial if these normalized representations are identical \cite{fpga24}). This symbolic representation is based on Computational Directed Acyclic Graphs (CDAGs \cite{elango2015characterizing}), which have a space complexity linear in the number of operations executed in the program. Illustrative examples and intuitions of this approach are detailed in \cite{fpga24}.
For example, in the context of High-Level Synthesis (HLS), we demonstrated that this approach can verify a highly optimized systolic array. We targeted an implementation produced by the AutoSA \cite{wang2021autosa} compiler consisting of 8k processing elements and 16k FIFOs. By successfully verifying this against a 5-line dense matrix-multiply loop, we demonstrated both scalability and coverage of key concurrency constructs.
Indeed, the correctness is proved by treating all elements of the input matrices $A,B$ as symbols. For every element of the output $C$, we compute a symbolic expression corresponding to the computation performed (itself only a function of elements of $A,B$ \cite{fpga24}). This symbolic analysis ensures equivalence will be satisfied whichever the actual concrete values these matrices may be at runtime. 

In this work, we entirely revisit the approach to concurrency support and parallelization correctness checking in PEQC, while significantly broadening the class of programs supported, including some input-dependent symbolic conditional control flow which is modeled as conditional CDAGs. 
Fig.~\ref{fig:flow1} outlines our system. 
The \PEQC components are presented in Sec.~\ref{sec:formal}, where we develop small-step operational semantics for a (simplified) parallel language supporting synchronizations using binary (\texttt{set}/\texttt{wait}) and counting (\texttt{acquire}/\texttt{release}) semaphores. We prove that our system only accepts deterministic parallel programs (reporting race/deadlock otherwise). We prove it computes a correct, equivalent symbolic representation of the programs for subsequent equivalence checking. 
This language is rich enough to accurately capture the semantics of a wide set of MLIR programs, after conversion and possible lowering, enabling the verification of pairs of MLIR programs including in the AIR and AIE dialects as well as standard dialects such as \texttt{async}, \texttt{scf}, and \texttt{affine}. This conversion is presented in Sec.~\ref{sec:mlirsupport} and is fundamental to deploy our staged approach to correctness verification. Experimental results on a variety of AIR and AIE designs for the AI Engine are presented in Sec.~\ref{sec:experiments}, including subtle concurrency bugs we have exposed for the first time, along with an extensive evaluation of \texttt{mlir-opt} optimizations over hundreds of benchmarks.

%% file: peqc-formal.tex
\newcommand{\expeval}[2]{#1 \Downarrow_{M, s} #2}
\newcommand{\expsym}[2]{#1 \Downarrow_{M, s} #2}

\newcommand{\ncode}[1]{#1_{code}}
\newcommand{\nid}[1]{#1_{id}}
\newcommand{\hb}{\rightsquigarrow}
\newcommand{\hbstar}{\hb^{*}}
\newcommand{\pstrans}[1]{\rightarrow^?_{#1}}
\newcommand{\strans}{\rightarrow}

\newcommand{\taskleaf}{task\_leaf}
\newcommand{\graphvalid}{graph\_valid}
\newcommand{\progeval}{\Downarrow}
\newcommand{\corrset}{corr\_set}
\newcommand{\corrrel}{corr\_rel}

\newcommand{\memupdate}{mem\_update}
\newcommand{\getvar}{get\_var}
\newcommand{\vareval}{var\_eval}
\newcommand{\memdecl}{mem\_decl}

We now formalize our approach to computing equivalence between programs, using a hybrid concrete-symbolic interpretation to compute a normalized symbolic representation of the final memory state when the program terminates without error.

We first present a simplified proxy for the parallel language we support in \PEQC in Sec.~\ref{subsec:grammar}, and how to evaluate and represent symbolic expressions in our model in Sec.~\ref{subsec:expreval}. We then present in Sec.~\ref{subsec:smallstepsem} a small-step operational semantics for this parallel language which is used to build a state-based representation of the program $(G,M)$, where $G$ is a graph with nodes representing each execution of every statement, and $M$ is the representation of the memory state. 
We prove our system can only reach a final state (i.e., the program terminates without error) for deterministic parallel programs, which implies the absence of race and deadlocks (and conversely, that our system errors in the presence of deadlocks and races), in Sec.~\ref{subsec:determinismproof}. 
In Sec.~\ref{subsec:equivalence}, we discuss proving equivalence between the representations obtained for two programs: $(G,M)\equiv(G',M')$, and present details of our full implementation in Sec.~\ref{sec:implementation}.

\subsection{A Simple Featherweight Language}
\label{subsec:grammar}

The grammar in Fig.~\ref{fig:grammar1} captures programs in our featherweight language. It is a simplified version of our intermediate representation implementation outlined in Sec.~\ref{subsec:implem-pir}, but it captures all critical constructs required to  build a symbolic representation of the program while reasoning on synchronizations with semaphores. The semantics of each construct is presented in Sec.~\ref{subsec:smallstepsem}.

\begin{center}
\begin{figure}[h!tb]
\vspace{-0cm}
{
$$
\begin{array}{l r c l}
\text{Program} & P & \rightarrow & s \text{; } P \mid \epsilon \\[6pt]
\text{Statement} & s & \rightarrow & \texttt{int } v[] \\
& & \mid & v[e_i] = e \\
& & \mid & \texttt{async} \{ P \} \\
& & \mid & \texttt{set}(e_1, e_2) \\
& & \mid & \texttt{wait}(e_1, e_2) \\
& & \mid & \texttt{acquire}(e_1, e_2) \\
& & \mid & \texttt{release}(e_1, e_2) \\
& & \mid & \texttt{while} (e) \{ P \} \\
\text{Expression} & e & \rightarrow & c \\
& & \mid & v[e_i] \\
& & \mid & e_1 + e_2 \\
& & \mid & e_1 < e_2 \\
& & \mid & (e_{cond} ? e_1 : e_2)
\end{array}
$$
}\vspace{-0cm}
\caption{\label{fig:grammar1}Grammar for the simplified language.}
\end{figure}
\end{center}

\subsection{Expression Evaluation}
\label{subsec:expreval}

Before describing the semantics of our featherweight but complete language, we first present the evaluation of expressions. This is key for our system: we mix concrete evaluation and symbolic representation for expressions, and rely on correctly implementing the concrete semantics of arithmetic expressions in the interpreter. For example, we implement the \texttt{int} 32-bit data type, with C semantics, to evaluate expressions composed of only concrete values. 
That is, we replace $1+1$ in the program with $2$ via concrete evaluation. But if values of some operands in the expression are not known, for example $in1+1$ where the input value $in1$ is read before being written in the program (live-in), we instead promote this expression to a symbolic representation using an AST-style tree representation: $+(in1, 1)$. 
This symbolic expression is then used to update other expressions. For example, after interpreting $in1 = in1+1;in1 = in1 *2$, $in1$ is represented as $sym(*(+(in1, 1),2))$, where $sym(e)$ denotes an expression $e$ represented symbolically.
As shown in Sec.~\ref{subsec:smallstepsem}, we restrict where symbolic expressions may occur during interpretation to ensure we deterministically obtain the final memory representation at the end of the program, for example preventing them in array subscripts (we must always be able to associate a unique name to a memory location) or in values set for semaphores.

Expression evaluation is defined in Fig.~\ref{fig:expression-rules}. $\expeval{e}{c}$ represents the evaluation of $e$ to $c$ (a concrete value) with respect to statement $s$. $err$ triggers an error state. Helper functions are defined in Sec.~\ref{subsec:smallstepsem}: $var\_eval(M,s,n,i)$ extracts the current representation in memory $M$ of the array cell named \texttt{n[i]}. 
Here we only illustrate key evaluation constructs; our full implementation develops this approach for additional unary and binary arithmetic operators.
\begin{figure}[h!tb]
    \centering
    
    \begin{minipage}{0.32\textwidth}
        \centering
        $$
        \inference{}{\expeval{\mathrm{c}}{c}}
        $$
    \end{minipage}
    \begin{minipage}{0.32\textwidth}
        \centering
        $$
        \inference{\expeval{e_i}{\textit{sym}(s)}}{\expeval{n[e_i]}{\textit{err}}}
        $$
    \end{minipage}
    \begin{minipage}{0.32\textwidth}
        \centering
        $$
        \inference{\expeval{e_i}{\textit{err}}}{\expeval{n[e_i]}{\textit{err}}}
        $$
    \end{minipage}
    
    \vspace{.5em}
    
    \begin{minipage}{0.48\textwidth}
        \centering
        $$
        \inference{\expeval{e_i}{i} & \vareval(M, s, n, i) = v}{\expeval{n[e_i]}{v}}
        $$
    \end{minipage}
    \begin{minipage}{0.48\textwidth}
        \centering
        $$
        \inference{\expeval{e_i}{i} & \vareval(M, s, n, i) = \bot}{\expeval{n[e_i]}{\textit{sym}(n[i])}}
        $$
    \end{minipage}
    
    \vspace{.5em}
    
    \begin{minipage}{0.48\textwidth}
        \centering
        $$
        \inference{\expeval{e_1}{v_1} & \expeval{e_2}{v_2}}{\expeval{e_1 + e_2}{\mathrm{val\_add}(v_1, v_2)}}
        $$
    \end{minipage}
    \begin{minipage}{0.48\textwidth}
        \centering
        $$
        \inference{\expeval{e_c}{v_c} & \expeval{e_1}{v_1} & \expeval{e_2}{v_2}}{\expeval{e_c ? e_1 : e_2}{\mathrm{val\_tern}(v_c, v_1, v_2)}}
        $$
    \end{minipage}
    
    \vspace{.5em}

    \begin{minipage}{0.48\textwidth}
        \raggedright 
        
        $\mathrm{val\_add}(v_1, v_2) =$
        $$
        \begin{cases}
        c_1 + c_2 & \text{if } v_1 = c_1, v_2 = c_2 \\
        \textit{sym}(s + c) & \text{if } v_1 = \textit{sym}(s), v_2 = c \\
        \textit{sym}(c + s) & \text{if } v_1 = c, v_2 = \textit{sym}(s) \\
        \textit{sym}(s_1 + s_2) & \text{if } v_1 = \textit{sym}(s_1), v_2 = \textit{sym}(s_2) \\
        \textit{err} & \text{if } v_1 = \textit{err} \lor v_2 = \textit{err}
        \end{cases}
        $$
    \end{minipage}
    {\footnotesize
    \begin{minipage}{0.48\textwidth}
        \raggedright
        
        $\mathrm{val\_tern}(v_c, v_1, v_2) =$
        $$
        \begin{cases}
        v_2 & \text{if } v_c = c \land c = 0 \\
        v_1 & \text{if } v_c = c \land c \neq 0 \\
        \textit{sym}(s_c ? c_1 : c_2) & \text{if } v_c = \textit{sym}(s_c), v_1=c_1, v_2=c_2 \\
        \textit{sym}(s_c ? s_1 : c) & \text{if } v_c = \textit{sym}(s_c), v_1=\textit{sym}(s_1), v_2=c \\
        \textit{sym}(s_c ? c : s_1) & \text{if } v_c = \textit{sym}(s_c), v_1=c, v_2=\textit{sym}(s_1) \\
        \textit{sym}(s_c ? s_1 : s_2) & \text{if } v_c = \textit{sym}(s_c), v_1=\textit{sym}(s_1), v_2=\textit{sym}(s_2) \\
        \textit{err} & \text{if } v_c = \textit{err} \lor v_1 = \textit{err} \lor v_2 = \textit{err}
        \end{cases}
        $$
    \end{minipage}
    }
\vspace{-.3cm}
    \caption{Expression evaluation semantics.}
    \label{fig:expression-rules}
    \vspace{-0.5cm}
\end{figure}

\subsection{Small-step Operational Semantics}
\label{subsec:smallstepsem}

Our objective is to reason on parallel programs with possibly hierarchical tasks spawned with \texttt{async} and synchronized with semaphores, proving the absence of data races and deadlock in a program or producing an error otherwise. 
While Jin et al. proved a race-free program is deterministic under synchronization with promises \cite{jin23}, here we prove that \emph{any} program is deterministically handled by our system, including when using counting semaphores. A correct parallel program must ensure that every operation that may execute in parallel cannot read/write the same memory location at the same time, and that it produces the same symbolic computation result as the reference program. Here we assume that only explicit synchronizations in the program (i.e., semaphore operations) can force an ordering between sets of statements in the program; we never rely on latency and timing of operations to determine their execution order.

\paragraph*{Notation} In our semantics, all variables are arrays of either symbolic or concrete values.
The memory of the program is represented by a function $M : (id, name) \rightharpoonup (\mathbb{N} \rightharpoonup value)$ that maps a tuple of task ID and variable name to an array of values.
$M$ need not be defined on a given task ID and variable name; similarly, a variable's value at any index need not be initialized.

We define several helper functions for accessing and modifying the memory representation.
The function $\getvar(M, id, v)$ returns the task ID where the variable name was declared: child threads can access their parents' variables.
The function $\vareval(M, id, v, i)$ returns the value of a variable at index $i$. 
If $\getvar(M,id, v) \in \mathrm{dom}(M)$ and $i \in \mathrm{dom}(M(\getvar(M,id, v)))$, then $\vareval(M, id, v, i) = M(\getvar(M,id, v))(i)$. Otherwise, $\vareval(M, id, v, i) = \bot$.

To modify $M$, we use the shorthand $M[(id, v) \mapsto (i \mapsto e)]$ to represent a function that is identical to $M$ except with the modification that $M(id, v)(i) = e$.
The function $\memdecl(M, id, v)$ declares a variable with a given task ID and name.
If $(id, v)$ is already in $\mathrm{dom}(M)$, $\memdecl(M, id, v) = M$.
Otherwise, $\memdecl(M, id, v) = M[(id, v) \mapsto \emptyset]$: initially, the variable is undefined at all indices.
The function $\memupdate(M, id, v, i, e)$ assigns the value of a variable at index $i$ to be a concrete or symbolic value, $e$:
$\memupdate(M, id, v, i, e) = M[\getvar(M, id, v) \mapsto (i \mapsto e)]$.

Intuitively, we present a system represented by a graph $G$ where every node is one instance of a statement, and where an edge between two nodes $s$ and $s'$ represent the requirement for $s$ to execute before $s'$, as per the program synchronizations.
Nodes in the graph are of the form $(id, P)$, where $id$ is the ID of the task executing the statement, and $P$ is the rest of the program the task will execute, starting with the statement the node represents.
We use $n \hb m$ to denote that an edge exists from $n$ and $m$ in $G$, and $n \hbstar m$ to denote that a \emph{path} exists from $n$ to $m$ in $G$.
A new node in $G$ is created for each statement instance. For example, in Figure~\ref{fig:formal-ex1}, nodes for statements in the loop body appear $N$ times in the graph.

\begin{definition}[Happens Before]
\label{def:hb}
For graph $G$ and two nodes $n, m \in G$, $n$ happens before $m$, denoted $n \hbstar m$, if and only if there exists a path from $n$ to $m$ in $G$. 
\end{definition}

\begin{definition}[May Happen in Parallel]
\label{def:mhp}
For graph $G$ and two nodes $n, m \in G$, $n$ may happen in parallel with $m$, denoted $n \parallel m$, if and only if $n \not\hbstar m$ and $m \not\hbstar n$.
\end{definition}

We define the semantics of programs as a relation over program states, where a state is either a tuple $(M, G)$ of program memory and graph, or an error state $error(G)$ that represents an invalid state. 
Error states are identical: for all $G, G'$, we consider $error(G)$ and $error(G')$ to be equivalent.

We present a very simple  example of a program and its graph $G$ in Fig.~\ref{fig:formal-ex1}. Our small-step semantics operates on the source code on the left to produce and manipulate the graph $G$ on the right. 
Note that we limit use to the featherweight language constructs, but the featherweight language concepts are capable of capturing core computations and their synchronizations. 
We briefly discuss in Sec.~\ref{subsec:implem-pir} the full set of additional operations and constructs we handle, and how to convert them to the featherweight language.

\begin{figure}[h!tb]
\begin{minipage}{0.99\textwidth}
\hspace{0.5cm}\begin{minipage}{0.49\textwidth}
\begin{lstlisting}
// all semaphores initialized to 0
i[0] = 0;
e[0] = i[0] < 2;
// for (i = 0; i < 2; i++)
while (e[0]) {
  async {
    A[i[0]] = A[i[0]] + 1;
    set(sems[i[0]], 1);
  }
  async {
    wait(sems[i[0]], 1);
    A[i[0]] = A[i[0]] * 2;
  }
  i[0] = i[0] + 1;
  e[0] = i[0] < 2;
}
\end{lstlisting}
\end{minipage}
\begin{minipage}{0.49\textwidth}
\centering
\hspace{-1cm}\includegraphics[width=0.9\textwidth]{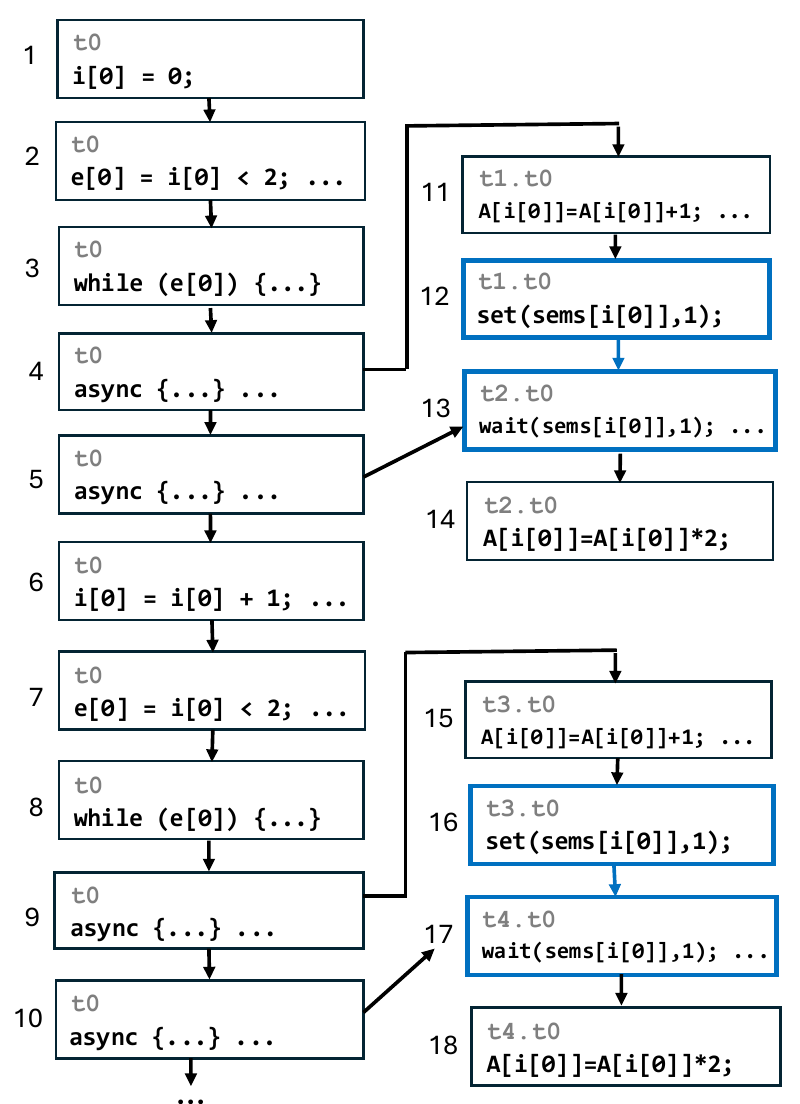}
\end{minipage}
\end{minipage}
\caption{\label{fig:formal-ex1}An example simple concurrent program, and the graph $G$ representing its happens-before relation. Nodes and edges in blue represent explicit synchronization.}
\end{figure}

Intuitively, any pair of nodes that may execute in parallel must be checked for modifications to conflicting memory locations.
Illustrated below, we develop a small-step operational semantics to build $G$ and $M$ for the program, and the associated checks to be performed to detect possible races. The small-step operational semantics describes constructively the scope of parallel programs we can interpret and reason about for equivalence, and those we report an error for.
\begin{figure}[h!tb]
    \centering
    
    \begin{minipage}{0.44\textwidth}
        \centering
        $$
        \inference
        {n = (s, \texttt{int }v[];P) \\ G' = G \cup \{n \hb (s, P)\} }
        {(M, G) \pstrans{n} (\memdecl(M, s, v), G')}
        $$
    \end{minipage}
    \begin{minipage}{0.54\textwidth}
        \centering
        $$
        \inference
        {n = (s, v[e_i] = e; P) & \expeval{e_i}{c} & \expeval{e}{e'} \\ G' = G \cup \{n \hb (s, P)\} }
        {(M, G) \pstrans{n} (\memupdate(M, s, v, c, e'), G')}
        $$
    \end{minipage}
    
    \vspace{0.5em}
    
    \begin{minipage}{0.44\textwidth}
        \centering
        $$
        \inference
        {n = (s, \texttt{while}(e)\{P'\};P) & \expeval{e}{c} \\ c = 0 & G' = G \cup \{n \hb (s, P)\} }
        {(M, G) \pstrans{n} (M, G')}
        $$
    \end{minipage}
    \begin{minipage}{0.54\textwidth}
        \centering
        $$
        \inference
        {n = (s, \texttt{while}(e)\{P'\};P) & \expeval{e}{c} & c \neq 0 \\ G' = G \cup \{n \hb (s, P';\texttt{while}(e)\{P'\};P)\} }
        {(M, G) \pstrans{n} (M, G')}
        $$
    \end{minipage}
    
    \vspace{0.5em}
    
    \begin{minipage}{0.48\textwidth}
        $$
        \inference
        {n = (s, \texttt{while}(e)\{P'\};P) \\ \expeval{e}{sym(sy)}\lor\expeval{e}{err}}
        {(M, G) \pstrans{n} error(G)}
        $$
    \end{minipage}
    \begin{minipage}{0.48\textwidth}
        $$
        \inference
        {n = (s, v[e_i] = e; P) \\ \expeval{e_i}{sym(sy) \lor \expeval{e_i}{err}} \lor \expeval{e}{err} }
        {(M, G) \pstrans{n} error(G)}
        $$
    \end{minipage}
    
    \vspace{0.5em}
    
    $$
    \inference
    {n = (s, \texttt{async}\{P'\};P) & id = fresh\_id(G) \\ G' = G \cup \{n \hb (s, P), n \hb (id \cdot s, P')\} }
    {(M, G) \pstrans{n} (M, G')}
    $$
    
    \vspace{0.5em}
    
    \begin{minipage}{0.39\textwidth}
        \centering
        $$
        \inference
        {n = (s, \texttt{set}(sem, e);P) \\ G' = G \cup \{n \hb (s, P)\} }
        {(M, G) \pstrans{n} (M, G')}
        $$
    \end{minipage}
    \begin{minipage}{0.59\textwidth}
        \centering
        $$
        \inference
        {n = (s, \texttt{wait}(sem, e);P) & \corrset(G, n, n_{set}) \\ G' = G \cup \{n \hb (s, P)\} \cup \{n_{set} \hb n\} }
        {(M, G) \pstrans{n} (M, G')}
        $$
    \end{minipage}
    
    \vspace{0.5em}
    
    \begin{minipage}{0.34\textwidth}
        \centering
        $$
        \inference
        {n = (s, \texttt{release}(sem, e);P) \\ G' = G \cup \{n \hb (s, P)\} }
        {(M, G) \pstrans{n} (M, G')}
        $$
    \end{minipage}
    \begin{minipage}{0.64\textwidth}
        \centering
        $$
        \inference
        {n = (s, \texttt{acquire}(sem, e);P) & \expeval{e}{c} \\ \sum\{r_{val} \mid \corrrel(G, n, r) \}  = c \\ G' = G \cup \{n \hb (s, P)\} \cup \{r \hb n \mid \corrrel(G, n, r)\} }
        {(M, G) \pstrans{n} (M, G')}
        $$
    \end{minipage}
    
    \vspace{0.5em}
    
    $$
    \inference
    { n = (s, \texttt{op}(sem, e);P) & \texttt{op} \in \{\texttt{set}, \texttt{wait}, \texttt{release}, \texttt{acquire}\} \\
      (\expeval{sem}{\textit{sym}(sy)} \lor \expeval{sem}{\textit{err}} \lor \expeval{e}{\textit{sym}(sy)} \lor \expeval{e}{\textit{err}})
    }
    {(M, G) \pstrans{n} \textit{error}(G)}
    $$

    \begin{minipage}{0.48\textwidth}
    $$
    \inference
    { 
    (M, G) \pstrans{n} (M', G') \\
    \taskleaf(G, n) & \graphvalid(G') 
    }
    {(M, G) \strans (M', G')}
    $$
    \end{minipage}
    \begin{minipage}{0.48\textwidth}
    $$
    \inference
    { 
    (M, G) \pstrans{n} (M', G') & \lnot\graphvalid(G')
    }
    {(M, G) \strans error(G')}
    $$
    \end{minipage}
    \vspace{0.5em}
    $$
    \inference
    { 
    (M, G) \pstrans{n} error(G')
    }
    {(M, G) \strans error(G')}
    $$

\vspace{-.3cm}
    \caption{Statement evaluation semantics.}
    \label{fig:statement-rules}
\vspace{-.5cm}
\end{figure}

Program semantics are defined in terms of two separate relations, shown in Fig.~\ref{fig:statement-rules}.
An intermediate relation, $\pstrans{} \subseteq State \times Node \times State$, defines potential transitions.
A state that potentially transitions from state $\sigma$ to $\sigma'$ by executing the first statement in node $n$ is denoted $\sigma \pstrans{n} \sigma$.
The relation $\rightarrow \subseteq State \times State$ represents the actual program semantics, a restriction of $\pstrans{}$ that only transitions on valid nodes.
We use $\strans^{*}$ to denote the reflexive transitive closure of $\strans$.

For a statement to be evaluated, it must be the most recent node added with its task ID.
This property is represented by a helper function $\taskleaf$: $\taskleaf(G, n)$ is true for node $n$ if and only if, for all other nodes $m \in G$, $n \hb m$ implies that $n$ and $m$ have different task IDs.
If $n \parallel m$, they must have different task IDs: statements are executed sequentially inside tasks, and different tasks always have different task IDs.

The rules in Figure~\ref{fig:statement-rules} build the state transition relation as follows.
The variable declaration rule declares a variable in the current task's scope using $\memdecl$.
The assign rule updates $M$ with a new value for the variable $v$ at the concrete index value $c$.
The subsequent two rules assign the next statements to be executed after a \texttt{while} statement according to the value of the condition $e$.
The two rules after these model errors due to symbolic values used in control flow or memory accesses.

The rest of the $\pstrans{}$ rules handle semaphore operations.
\PEQC and its semantics support two types of semaphores: counting semaphores, which use the operations \texttt{acquire} and \texttt{release}, and binary semaphores, which use the operations \texttt{set} and \texttt{wait}.
All of these functions take two arguments: $sem$, the ID of the semaphore the operation is being performed on, and a concrete value $v$.
Both semaphores are restricted in their operation to maintain determinism, and only manipulate the happens-before graph $G$ rather than maintaining an internal counter.

Counting semaphores in \PEQC mimic the semaphores typical to operating systems.
The operation  \texttt{release}$(sem,v)$ is a signal that increments the value of $sem$ by $v$, and the \texttt{acquire}$(sem,v)$ operation blocks the execution of the task that invokes it until the value of $sem \geq v$, then decrements $sem$ by $v$.
These operations are represented through graph operations which match \texttt{release} and \texttt{acquire} nodes.
For \PEQC to evaluate programs deterministically, we restrict \texttt{release}s: they can be matched to only one \texttt{acquire}, which requires that they sum exactly to the \texttt{acquire} value $v$ rather than being greater than or equal to $v$.

For binary semaphores in \PEQC, 
on a call to \texttt{wait}$(sem, v)$, a binary semaphore $sem$ pauses the execution of its task until another task calls \texttt{set}$(sem, v)$.
Similarly to counting semaphores, to ensure determinism, we restrict the use of binary semaphores such that 
two concurrent calls to \texttt{set} which correspond to one \texttt{wait} are disallowed.

Figure~\ref{fig:formal-ex1} shows an example of binary semaphores. Inside the loop body, \texttt{wait} operations ensure an ordering exists between the two statements that access \texttt{A[i[0]]}, because threads $(t2 \cdot t0)$ and $(t4 \cdot t0)$ cannot proceed until after threads $(t1 \cdot t0)$ and $(t3 \cdot t0)$, respectively, call \texttt{set}.

To maintain the correct happens-before relation over statements, we need to explicitly add synchronization from semaphores to the happens-before graph.
A semaphore \texttt{wait} must happen after its corresponding \texttt{set}, and a semaphore \texttt{acquire} must happen after its corresponding \texttt{release}s.
To add these edges to $G$, we define helper functions that define whether a \texttt{set} and \texttt{release} correspond to a \texttt{wait} and \texttt{acquire}, respectively.

For a node $n$ where $\ncode{n} = \texttt{op}(sem, e)$ such that $\texttt{op} \in \{\texttt{set}, \texttt{wait}, \texttt{release}, \texttt{acquire}\}$, we say that $n_{sem} = v_s$ if $\expeval{sem}{v_s}$ and that $n_{val} = v_e$ if $\expeval{e}{v_e}$.
We define a helper function $\corrset$ to calculate whether a \texttt{set} corresponds to a \texttt{wait}.
$\corrset(G, w, s) = true$ if and only if $w$ is a \texttt{wait} in $G$, $s$ is a \texttt{set} in $G$, and $w_{sem} = s_{sem}$ and $w_{val} = s_{val}$, and $s$ happens before or in parallel with $w$: $(w \parallel s \lor s \hb w)$.
Additionally, $s$ must be the most recent \texttt{set} with respect to $w$: there must exist no \texttt{set} $s' \in G$ such that $(w_{sem} = s'_{sem} \land s \hb s' \land (w \parallel s' \lor s' \hb w))$.
For example, in Figure~\ref{fig:formal-ex1}, $\corrset(G, n_{13}, n_{12}) = \corrset(G, n_{17}, n_{16}) = true$.

Similarly, we define a helper function $\corrrel$ for \texttt{release}s and \texttt{acquire}s.
$\corrrel(G, a, r) = true$ if and only if $a$ is an \texttt{acquire} in $G$, $r$ is a \texttt{release} in $G$, $a$ and $r$ operate on the same semaphore ($a_{sem} = r_{sem}$), $r$ happens before or in parallel with $a$ ($a \parallel r \lor r \hb a$), and there does not exist an acquire $a'$ on the same semaphore such that $r \hb a' \hb a$.
Unlike binary semaphores, counting semaphore \texttt{release}s increment the semaphore rather than overwrite it, so there is no need to restrict to the ``most recent'' \texttt{release}.

An important caveat to these semantics is that one semaphore ID can only be used as \emph{one} type of semaphore: programs must not, for example, contain \texttt{set}$(s, e)$ and \texttt{acquire}$(s', e)$, where $s$ and $s'$ evaluate to the same value.

\subsubsection{Graph Validity}

The relation $\rightarrow$ also requires that, for an edge $(M, G) \strans (M', G')$ to exist, $G'$ must be \emph{valid}.
Graph validity is defined in the helper function $\graphvalid$ below.

In state $(M, G)$, a node $n = (v[e_i] = e; P)$ is a \emph{write} to $v[c]$ if $\expeval{e_i}{c}$.
A node $n = (v[e_i] = e; P)$ is a \emph{read} from $v'[c]$ if there exists a subtree of the form $v'[e_i']$ such that $\expeval{e_i'}{c}$.
A node $n$ is an \emph{access} to $v[c]$ if it is either a read from or a write to $v[c]$.
Note that nodes containing semaphore operations can read from variables but not write to them.

For a graph $G$, $\graphvalid(G)$ is true if and only if:

\begin{enumerate}

    \item \label{gv:race} No two accesses, where one is a write, may happen in parallel: there exist no two nodes $n, m \in G$ such that $n \parallel m$, and there exists no location $v[c]$ where both $n$ and $m$ access $v[c]$ and either $n$ or $m$ is a write to $v[c]$.
    
    \item \label{gv:double_set} No \texttt{wait} has more than one corresponding \texttt{set}: there exists no \texttt{wait} $w \in G$ such that, for two \texttt{set}s $s, s' \in G$, $\corrset(G, w, s) \land \corrset(G, w, s')$.

    \item \label{gv:double_acquire} No \texttt{release} has more than one corresponding \texttt{acquire}: there exists no \texttt{release} $r \in G$ such that, for two \texttt{acquire}s $a, a' \in G$, $\corrrel(G, a, r) \land \corrrel(G, a', r)$.
    
    \item \label{gv:acquire_par} No \texttt{acquire}s to the same semaphore may happen in parallel: there exist no two \texttt{acquire}s $a, a' \in G$ such that $a_{sem} = a'_{sem} \land a \parallel a'$.

    \item \label{gv:acquire_sum} The sum of an \texttt{acquire}'s corresponding \texttt{release} values must not be greater than the value of the \texttt{acquire}: there exists no \texttt{acquire} $a \in G$ such that $\sum\{r_{val} \mid \corrrel(G, a, r)\} > a_{val}$.
    
\end{enumerate}

For the example graph $G$ in Fig.~\ref{fig:formal-ex1}, $\graphvalid(G) = true$, but if the semaphore synchronization with nodes 12, 13, 16, and 17 were not present, there would be no edges between nodes 11/14 and 15/18 that access the same index of \texttt{A}, which would cause the graph to be invalid due to a data race.

\subsubsection{Program Evaluation}

To evaluate a program $P$, we first assign it an initial state $(M^P_0, G^P_0)$. 
The starting graph, $G^P_0$, contains one node that assigns $P$ to a single initial task with ID $id_0$: $G^P_0 = \{(id_0, P)\}$. 
The starting memory state, $M^P_0$, has no variables defined: $\mathrm{dom}(M^P_0) = \emptyset$.

We define a program evaluation relation $\progeval$ such that $P \progeval \sigma$ if and only if $(M^P_0, G^P_0) \strans^{*} \sigma$ and $\sigma \not\strans \sigma'$ for any state $\sigma'$. 
We call $\sigma'$ a \emph{final state}.

\subsection{Determinism in Programs}
\label{subsec:determinismproof}

To prove equivalence between two programs using their final states, we first need to prove that a final state is \emph{unique}, i.e., that for every program $P$, only one state $\sigma$ exists such that $P \progeval \sigma$.
This means that program interpretation is deterministic, i.e., will always produce the same end state for every program.
To prove the determinism of our semantics, we follow a similar approach to prior work on deterministic concurrent systems \cite{jin23, budimlic10} and prove confluence of our semantics.

For any state $\sigma$, let $\sigma.G$ be the state's memory graph $G$, whether the state is of the form $(M, G)$ or $error(G)$.
Similarly, let $\sigma.M$ be the (non-error) state's memory $M$.
We use $\sigma \strans_n \sigma'$ as a shorthand for $\sigma \pstrans{n} \sigma' \land \sigma \strans \sigma'$.

\begin{lemma}[Weak Progress]
\label{lem:wprog}
Let $\sigma$ be a state where $\sigma_0 \rightarrow^{*} \sigma$, where $\sigma_0$ is some starting state.
If $\sigma_0 \strans_n \sigma_n$ and $\sigma_0 \strans_m \sigma_m$, then either $\sigma_m = error$ or there exists a state $\sigma_{mn}$ such that $\sigma_0 \strans_m \sigma_m \strans_n \sigma_{mn}$.
\end{lemma}
\begin{proof}
All transitions $\rightarrow_m$ must only add edges \emph{to} nodes with the same scope as $m$.
We know that $n$ and $m$ have different task IDs because $n \parallel m$, so we can conclude that $n$ will still be a task leaf in $\sigma_m.G$.
Since $n$ is still a task leaf, the only way for $\strans_n$ to not be a valid transition out of $\sigma_m$ is if $n$ is a \texttt{wait} or \texttt{acquire} and $\sigma_m.G$ does not contain the necessary \texttt{set} or \texttt{release} nodes.
We know that $\sigma_0.G$ contains the necessary nodes and that no transitions remove nodes from the graph, so we can conclude that $\sigma_m \strans \sigma_{mn}$ for some $\sigma_{mn}$.
\end{proof}

\begin{lemma}[Graph Independence]
\label{lem:gindep}
For nodes $n$ and $m$, if $\sigma_0 \rightarrow_n \sigma_n$ and $\sigma_0 \rightarrow_m \sigma_m \rightarrow_n \sigma_{mn}$, then $\sigma_{mn} = error$ or $\sigma_{mn}.G \setminus\sigma_m.G = \sigma_n.G \setminus \sigma_0.G$.
If $\sigma_n = error$, then $\sigma_{mn} = error$.
\end{lemma}
\begin{proof}

We consider two different cases based on the rule that produced the transition $\sigma_0 \rightarrow_n \sigma_n$.

First, we consider the cases where $\sigma_0 \pstrans{n} error(\sigma_0.G)$.
This is the only case where $\sigma_n = error$ and $\graphvalid(\sigma_n.G)$, and corresponds to a \texttt{while} condition, a variable index containing a symbolic value, or a semaphore operation depending on a symbolic value.

In the case where $\sigma_n = \sigma_{mn} = error$, the lemma trivially holds. 

We show by contradiction that a case where $\sigma_n = error$ and $\sigma_{mn} \neq error$ cannot exist. 
Since $\graphvalid(\sigma_n.G)$, there must exist a symbolic variable, $v[c]$, which is read in a loop condition expression (for \texttt{while}), an index expression (for assignments), or an argument (for semaphore operations) in $n$.
This means that $v[c]$ has not been written to in $\sigma_0$.
For $\sigma_{mn}$ to not be $error$, $m$ must be a write to $v[c]$ to make $v[c]$ defined in $\sigma_m.M$.
Since $n \parallel m$, this violates item~\ref{gv:race} of $\graphvalid$ causing a race condition, so $\lnot\graphvalid(\sigma_{mn}.G)$.
Therefore, $\sigma_{mn} = error$, contradicting the assumption that $\sigma_{mn} \neq error$.

Next, we consider the cases where $\sigma_0 \pstrans{n} \sigma_n$, where either $\graphvalid(\sigma_n.G)$ and $\sigma_n \neq error$, or $\lnot\graphvalid(\sigma_n.G)$ and $\sigma_0 \strans error$.
In both cases, $\sigma_m \neq error$.

If $n = (s, p;P)$ and $p$ is not an \texttt{async}, \texttt{while}, \texttt{wait}, or \texttt{acquire} statement, its evaluation must add only $\{n \rightarrow (s, P)\}$ to the graph according to the semantics, regardless of the previous memory or graph state. 
If $m$ is evaluated first, we know that $n$ can be evaluated from Lemma~\ref{lem:wprog} and that it will add the same node.
In this case, $\sigma_{mn}.G \setminus\sigma_m.G = \sigma_n.G \setminus \sigma_0.G = \{n \rightarrow (s, P)\}$.
We also must prove that $\sigma_n = error$ implies that $\sigma_{mn} = error$.
Since $\graphvalid(\sigma_m.G)$ and $\graphvalid(\sigma_0.G)$, the invalidity of $\sigma_n.G$ must result from the additions from $n$, which $\sigma_{mn}$ also contains.
The only way additions from $m$ could fix this invalidity is if $n$ causes a race from $\graphvalid$ item~\ref{gv:race} and $m$ added an edge that synchronized $n$ with the other statement(s) causing a race.
This is impossible because all transitions from $m$ can only add edges \emph{to} nodes with the same ID, and $n \parallel m$.

If $n = (s, \texttt{async}\{P'\};P)$, the reasoning is the same as above: the nodes added to the graph are independent of the memory state and $n$ must be a task leaf in both $\sigma_0$ and $\sigma_m$, so $\sigma_{mn}.G \setminus\sigma_m.G = \sigma_n.G \setminus \sigma_0.G = \{n \rightarrow (s, P), n \rightarrow (id \cdot s, P')\}$.

If $n = (s, \texttt{wait}(sem, e);P)$, 
there exists a node $n_s$ in $\sigma_0.G$ such that $\corrset(\sigma_0.G, n, n_s)$ because there exists a transition $\rightarrow_n$ from $\sigma_0$, and that $\lnot \corrset(\sigma_0.G, n, m)$ because $n \parallel m$ and $\graphvalid(\sigma_0.G)$. 
If $\sigma_m.G \setminus \sigma_0.G$ contains a node $m'$ such that $\corrset(\sigma_0.G, n, m')$, then $\sigma_{mn} = error$ because $n$ has more than one corresponding \texttt{set}, violating $\graphvalid$ item~\ref{gv:double_set}.
Otherwise, the \texttt{wait} can only add the following set of edges: $\sigma_{mn}.G \setminus\sigma_m.G = \sigma_n.G \setminus \sigma_0.G = \{n \rightarrow (s, P), n_s \rightarrow n\}$.

Similarly, if $n = (s, \texttt{acquire}(sem, e);P)$, there must exist a set of nodes $N_r \subseteq \sigma_0.G$ such that $\corrrel(\sigma_0.G, n, n_r)$ for all $n_r \in N_r$, and $|N_r| = c$ where $\expeval{e}{c}$ because there exists a transition $\rightarrow_n$ from $\sigma_0$.
If $\sigma_m.G \setminus \sigma_0.G$ contains a node $m'$ such that $\corrrel(\sigma_0.G, n, m')$, then $\sigma_{mn} = error$ because the sum of corresponding \texttt{release} values is greater than $c$, violating $\graphvalid$ item~\ref{gv:acquire_sum}.
If $\sigma_m.G \setminus \sigma_0.G$ contains a node $m'$ such that $m'$ is an \texttt{acquire} and $m' \parallel n$, then $\sigma_{mn} = error$ because two \texttt{acquire}s can happen in parallel, violating $\graphvalid$ item~\ref{gv:acquire_par}.
Otherwise, the \texttt{acquire} can only add the following set of edges: $\sigma_{mn}.G \setminus\sigma_m.G = \sigma_n.G \setminus \sigma_0.G = \{n \rightarrow (s, P)\} \cup \{n_r \rightarrow n \mid n_r \in N_r\}$.

Statements that use \texttt{while} are the only statements where the nodes that are produced by evaluation depends on the value of an expression. 
If $n = (s, \texttt{while}(e)\{P'\};P)$, we first consider the case where the set of variables read in the condition $e$ and the set of variables written to by $m$ are disjoint.
In this case, the values of $e$ when evaluated under $\sigma_n.M$ and $\sigma_{mn}.M$ are the same because $m$ does not write to any variables read in $e$, so the node added to the graph will be the same in $\sigma_n$ and $\sigma_{mn}$.
In the case where $m$ writes to a variable read in the condition $e$, then we know $\lnot\graphvalid(\sigma_{mn}.G)$ as per $\graphvalid$ item~\ref{gv:race}, so $\sigma_{mn}.G = error$.
\end{proof}

For two memory states $M$ and $M'$, we define a difference $M' - M$ that represents the changes to $M$ performed in $M'$.
If $(id, v) \in \mathrm{dom}(M')$ and either $(id, v) \not\in \mathrm{dom}(M)$ or $(id, v) \in \mathrm{dom}(M) \land M(id, v) \neq M'(id, v)$, $(M' - M)(id, v) = M'(id, v)$. 
We need to prove that changes to the memory state performed by statements that may happen in parallel are disjoint, which we formulate similarly to Lemma~\ref{lem:gindep}.

\begin{lemma}[Memory Independence]
\label{lem:mindep}
For nodes $n$ and $m$, if $\sigma_0 \rightarrow_n \sigma_n$ and $\sigma_0 \rightarrow_m \sigma_m \rightarrow_n \sigma_{mn}$, then $\sigma_{mn} = error$ or $\sigma_{mn}.M -\sigma_m.M = \sigma_n.M - \sigma_0.M$.
\end{lemma}
\begin{proof}

We only consider $\sigma_{mn} \neq error$, where $\graphvalid(\sigma_{mn}.G)$ (otherwise the lemma trivially holds).
If $n$ is not a declaration or an assignment, $\sigma_{mn}.M -\sigma_m.M = \sigma_n.M - \sigma_0.M = \emptyset$ since no other transitions can modify the memory state.

If $n_{code} = (\texttt{int }v[];P)$, we only need to consider the case where $(\nid{n}, v) \not\in \mathrm{dom}(\sigma_0.M)$, since otherwise the transition $\strans_n$ does not modify the memory state.
We assume, for the purpose of contradiction, that $(\nid{n}, v) \in \mathrm{dom}(\sigma_m.M)$.
In this case, $m$ must be either a declaration or an assignment.
If $(\nid{n}, v)$ has not already been declared, the only way $m$ can access $(\nid{n}, v)$ is if $\nid{m} = \nid{n}$, because otherwise $\getvar$ will only return a variable in the global scope, not the variable associated with $\nid{n}$.
Since $n \parallel m$, $n$ and $m$ must have different task IDs, which contradicts our assumption.
We now know that $(\nid{n}, v) \not\in \mathrm{dom}(\sigma_m.M)$, so we can conclude that
$\sigma_{mn}.M -\sigma_m.M = \sigma_n.M - \sigma_0.M = \{((\nid{n}, v), \emptyset)\}$ from the definition of $\memdecl$.

If $n_{code} = (v[e_i] = e;P)$, we first consider the case where the set of variables read in the expressions $e$ and $e_i$ and the set of variables written to by $m$ are disjoint.
In this case, the values of $e$ and $e_i$ when evaluated under $\sigma_0.M$ and $\sigma_m.M$ are the same because $m$ does not write to any variables read in $e$, so $n$ will make the same modifications.
In this case, $\sigma_{mn}.M -\sigma_m.M = \sigma_n.M - \sigma_0.M = \{((\nid{n}, v), (c_i \mapsto e')\}$, where $\expeval{e_i}{c_i}$ and $\expeval{e}{e'}$.
In the case where $m$ writes to a variable read in the condition $e$, then we know $\lnot\graphvalid(\sigma_{mn}.G)$ because $n \parallel m$ as per $\graphvalid$ item~\ref{gv:race}, so $\sigma_{mn}.G = error$.
\end{proof}

\begin{theorem}[Diamond Property]
\label{th:diamond}
For state $\sigma = (M, G)$, if $\sigma \rightarrow_n \sigma'$ and $\sigma \rightarrow_m \sigma''$ then there exist $\sigma_c', \sigma_c''$ such that $\sigma' \rightarrow^{=} \sigma_c'$ and $\sigma'' \rightarrow^{=} \sigma_c''$ and $\sigma_c' = \sigma_c''$, where $\rightarrow^{=}$ is the reflexive closure of $\rightarrow$.
\end{theorem}

\begin{proof}
First, we know that $\sigma \neq error$ and that $\graphvalid(G)$. 
If $\sigma$ were $error$, there would be no outgoing transitions from $\sigma$.
$G$ must be valid because the starting graph for every program is valid, and all transitions preserve graph validity except transitions to $error$.
We also know $n \parallel m$. If $n$ and $m$ were ordered, only one would be a leaf and thus only one would be evaluable from $\sigma$.

If both $\sigma'$ and $\sigma''$ are $error$, $\sigma_c' = \sigma_c'' = error$.

If exactly one of $\sigma'$ and $\sigma''$ is $error$ (we assume wlog that $\sigma' = error$ and $\sigma'' \neq error$), $\sigma_c' = \sigma_c'' = error$.
From Lemma~\ref{lem:wprog}, we know that a transition $\sigma'' \rightarrow_n \sigma_c$ exists for some $\sigma_c$.
From Lemma~\ref{lem:gindep}, we know that $\sigma' = error$ implies $\sigma_c = error$, so $\sigma_c' = \sigma_c'' = error$.

If neither $\sigma'$ nor $\sigma''$ is $error$, then we pick $\sigma_c'$ and $\sigma_c''$ such that $\sigma' \rightarrow_m \sigma_c'$ and $\sigma'' \rightarrow_n \sigma_c''$ (we know that these states must exist from Lemma~\ref{lem:wprog}).
From Lemma~\ref{lem:gindep}, we know that $\sigma_c'.G \setminus \sigma'.G = \sigma''.G \setminus \sigma.G$ and $\sigma_c''.G \setminus \sigma''.G = \sigma'.G \setminus \sigma.G$.
We also know that $\sigma.G \subseteq \sigma'.G \subseteq \sigma'_c.G$ and $\sigma.G \subseteq \sigma''.G \subseteq \sigma''_c.G$, because all transitions can only add edges to the graph, not remove them.
From this, we know that $\sigma_c'.G = \sigma_c''.G = \sigma'.G \cup \sigma''.G$, and we define $G_c = \sigma'.G \cup \sigma''.G$.

If $\lnot\graphvalid(G_c)$, $\sigma_c' = \sigma_c'' = error$.
Otherwise, from Lemma~\ref{lem:mindep}, we know that $\sigma_c'.M - \sigma'.M = \sigma''.M - \sigma.M$ and $\sigma_c''.M - \sigma''.M = \sigma'.M - \sigma$.
Since there are no transitions $(M, G) \strans (M', G')$ such that $M \not\subseteq M'$, we know that $\sigma.M \subseteq \sigma'.M \subseteq \sigma'_c.M$ and $\sigma.M \subseteq \sigma''.M \subseteq \sigma''_c.M$.
From this, we know that $\sigma_c'.M = \sigma_c''.M = \sigma.M \cup (\sigma'.M - \sigma.M) \cup (\sigma''.M - \sigma.M)$.
We can now conclude that $\sigma_c' = \sigma_c''$.
\end{proof}

\begin{theorem}[Confluence]
\label{th:confluence}
For state $\sigma$, if $\sigma \rightarrow^{*} \sigma'$ and $\sigma \rightarrow^{*} \sigma''$ then there exist states $\sigma_c', \sigma_c''$ such that $\sigma' \rightarrow^{*} \sigma_c'$, $\sigma'' \rightarrow^{*} \sigma_c''$, and $\sigma_c' = \sigma_c''$.
\end{theorem}
\begin{proof}
Semi-confluence follows from Theorem \ref{th:diamond} via induction on the length of the path to $\sigma'$, from which confluence follows via induction on the length of the path to $\sigma''$ \cite[\S~2.7]{baader98}.
\end{proof}

\begin{theorem}[Determinism]
\label{th:determinism}
If $P \Downarrow \sigma$ and $P \Downarrow \sigma'$, $\sigma = \sigma'$. 
\end{theorem}
\begin{proof}
Let the initial state of $P$ be $\sigma_0$. 
From the definition of $\progeval$, $\sigma_0 \strans^{*} \sigma$ and $\sigma_0 \strans^{*} \sigma'$.
From Theorem \ref{th:confluence}, we know that there exists a $\sigma_c'$ such that $\sigma \rightarrow^{*} \sigma_c'$ and a $\sigma_c''$ such that $\sigma' \rightarrow^{*} \sigma_c''$. 
Since both $\sigma$ and $\sigma'$ are final states, $\sigma = \sigma_c' = \sigma' = \sigma_c''$.
\end{proof}

Determinism also means that deadlock is always detectable in our system.
We define deadlock as a non-error state with no outgoing transitions (that is "stuck") where there exists a node that is a task leaf, but has not finished executing its corresponding program.

\begin{definition}[Deadlock]
\label{def:deadlock}
A state $\sigma = (M, G)$ is deadlocked if it is a final state and there exists a node $n$ in $G$ such that $\taskleaf(n)$ and $\ncode{n} \neq \epsilon$.
\end{definition}

\begin{corollary}[Deadlock Detection]
If $P \Downarrow \sigma$ and $\sigma$ is deadlocked, every final state of $P$ deadlocks.
\end{corollary}
\begin{proof}
Follows from Theorem \ref{th:determinism}.
\end{proof}

\subsection{Equivalence between Programs}
\label{subsec:equivalence}

Our final theorem defines equivalence between programs. First, if two programs $P_A$ and $P_B$ have reached their final state without error, and the memory computed for each is identical, then they are equivalent, denoted $P_A \equiv P_B$.

\begin{theorem}[Strict program equivalence]
\label{th:equiv1}
Suppose that $P_A \Downarrow \sigma_A =(M_A, G_A)$, $P_B \Downarrow \sigma_B =(M_B, G_B)$, $\sigma_A \neq error$, and $\sigma_B \neq error$. If $M_A = M_B$, then $P_A \equiv P_B$.
\end{theorem}
\begin{proof}
The final state is unique per Theorem~\ref{th:determinism}. If $M_A = M_B$, then for all $v$, we know $ M_A(v) = M_B(v)$. If all variables accessed by both programs hold the same final value, then $P_A$ computes the same results as $P_B$; that is, $P_A \equiv P_B$.
\end{proof}

Note that we can broaden beyond strict equality by designing a system to normalize $M$ before checking for equality of values, for example by deploying rewrite rules to normalize symbolic expressions (e.g., $sym(a*b) \equiv sym(b*a)$), as discussed below. 

Technically, Theorem~\ref{th:equiv1} is overly restrictive: only variables which are still visible when the program terminates need to have the same values, while temporary variables (e.gw., a loop counter) may have different values. In our implementation, we keep track of local variables (those local to tasks spawned by the program and therefore deleted when the task terminates without error), and only check equivalence for the remaining non-local variables. As we identify variables by name across programs, it is necessary for non-local variables to have identical names in both programs, a typical scenario when one program is a transformation of the other. 
This leads to Theorem~\ref{th:equiv2}, where we check the union of programs' non-local variables to detect side effects appearing in one program but not the other.

\begin{theorem}[Equivalence of non-local variables]
\label{th:equiv2}
Suppose $P_A \Downarrow \sigma_A = (M_A, G_A)$ and $P_B \Downarrow \sigma_B =(M_B, G_B)$, and that
$S = \{ v \mid M_A(v) = M_B(v) \}$ is the set of variables which are equivalent. Suppose that $V$ is the union of the non-local variables in $P_A$ and $P_B$: memory locations which are reachable after program termination. Then $V \subseteq S$ implies that $P_A \equiv P_B$.
\end{theorem}
\begin{proof}
Follows from Theorem~\ref{th:equiv1} and Theorem~1 in \cite{fpga24} proving the final CDAGs can only contain symbols of non-local variables.
\end{proof}

We remark that equivalence only holds for the set $V$ in Theorem~\ref{th:equiv2}. $V$ may contain concrete values for some variables (e.g., \texttt{N=32}) and symbolic expressions for other
(e.g., their final value is stored as a CDAG 
\texttt{C[0][0]=...+(+(*(A[0][0],} \texttt{B[0][0]),*(A[0][1],B[1][0]),...)}).
If the programs are correct executable tests, all values are concrete in $V$, and the proof holds only for these concrete values. However, programs can also contain input non-local variables that are not written before being read by the program, which are automatically treated as symbolic. The equivalence result holds for any value symbolic variables may concretely take. However, for our hybrid interpretation to terminate, the control-flow needs to be deterministically computed. Typically, this requires fixed problem sizes (e.g., \texttt{N=32}). As shown in Sec.~\ref{sec:experiments}, for a large class of useful programs such as dense linear algebra and related AI workloads, only the problem size is fixed; the equivalence computation is performed on symbolic input data. This makes the correctness result highly practical and far more general than simple testing.

%% file: peqc-implementation.tex
We now present additional details on our full \PEQC implementation for correctness checking, before presenting how to translate and lower MLIR programs to \PEQC's IR (PIR) in Sec.~\ref{sec:mlirsupport}.

\subsection{From Featherweight language to PIR}

We proved in the prior Sec.~\ref{sec:formal} that any program processed to termination without error by our system is a deterministic concurrent program, free of races and deadlocks, and amenable to simple equivalence checking with other programs. Technically, this is restricted to two programs where one is a substitute for the other, i.e., they both have necessarily the same execution context. Details are provided in \cite{fpga24}. 

The featherweight language carries all the necessary key constructs (synchronizations, parallelism, control-flow), but it remains a subset of the complete language we support.
\PEQC supports a subset of the C language in terms of control-flow and arithmetic operations. They can either be translated immediately to the featherweight language, or are treated identically to constructs in it.

The \PEQC intermediate representation, that is, the language that \PEQC takes as input, contains all standard unary and binary operators from the C language, whose treatment is analogous to $e_1+e_2$. 
Blocks, for loops, and if-conditionals can be reconstructed from $\text{while} (e) \{ P \}$. 
Function calls require that user guarantees the function is pure (e.g., $min(a,b)$) or that the function definition exists in the file; control is transferred to its definition, possibly after assigning arguments. 
Goto statements are handled similarly.
We support general \texttt{for (init; test; increment) body;} loops in C, they can be immediately translated to \texttt{init; while (test) \{ body; increment;\}}; similarly for \texttt{if (e) \{...\}}. Function calls are processed using an inlining-style approach, and recursive functions are not supported. We support all operators of the C language (\texttt{+}, \texttt{*}, \texttt{/}, \texttt{!=}, \texttt{>{}>}, etc.) in the concrete evaluator, following C semantics.

In addition, we provide an API (functions and types) intercepted by \PEQC to represent concurrency constructs: \texttt{peqc\_semaphore\_wait()}, etc. for each of the operations in the featherweight language (\texttt{async}, \texttt{wait}, \texttt{set}, \texttt{acquire}, \texttt{release}) which behave as formalized in the prior section. We also provide a specific data type for the semaphores, enabling their tracking in the program easily.

PIR is rich enough to capture a rich set of concurrent programs, as exemplified in Sec.~\ref{sec:experiments}. Our prior work includes supporting a subset of OpenMP pragmas \cite{tucker2025verification} which are lowered to PIR using semantically equivalent \texttt{async} and semaphore synchronizations. Details and documentation are available directly in the \PEQC software on Github \cite{peqc-mlir-github}.

\label{subsec:implem-pir}

\subsection{Hybrid Interpretation and Concurrency Checking}

During interpretation, the happens-before graph $G$ is represented explicitly, with a node for each statement instance.
In practice, the size of this graph can be reduced significantly by combining nodes in $G$, reducing its size and improving the speed to query it for correctness checking.

If two statement instances do not contain semaphore operations or spawn a task, both nodes will only have one input and output edge.
The nodes corresponding to these statements can therefore be merged into a single "macro-node" during interpretation, to reduce the size of the graph and thus reduce the time to check the happens-before relationship between two nodes. We implemented this optimization in \PEQC.

\subsection{Equivalence Checking}

As shown in Sec.~\ref{sec:experiments}, the time to conduct equivalence checking is mostly negligible compared with the time to build the normalized memory representations using CDAGs \cite{fpga24}: we typically rely exclusively on exact matching between the two memory states: that is, computing whether CDAGs are identical, a simple linear-time process. To reduce memory space, we employ memory pooling and shadow copies, so that CDAGs used in multiple live-out memory cells are represented only once in memory, leading to linear-time equivalence checking. However, there are situations where the user may wish to deploy further normalization, for example to tolerate associative/commutative reorderings of operations, as described below.

\paragraph*{Normalization of symbolic expressions} 
\PEQC includes a rewrite-rule system that can be applied on CDAGs to normalize them prior to equivalence checking. This proves particularly useful to handle transformations that change the nature and count of symbolic operations, such as factorization/distribution for expressions. The user must ensure that these are semantically applicable on the program verified; we do not perform any such analysis. 
We did not use nor need any such custom rewrites in our experiments in Sec.~\ref{sec:experiments}.

We also support normalization of associative and commutative operations by sorting the tree lexicographically (average $O(n log(n))$) after merging associative operations under a $n$-ary node with all operations of that type, supporting \texttt{+,*,min,max} natively, similar to Eker et al. \cite{eker03}. 
We also offer simple normalizations, such as \texttt{/(\$x \$x) = 1}, \texttt{+(\$x 0) = \$x}, \texttt{*(\$x 0) = 0}, as well as their commuted variants.

%% file: mlirsupport.tex
We now present our approach to lowering and translation of a rich class of MLIR programs to our parallel language intermediate representation described in the previous section. 

Our approach to support MLIR program correctness verification is to translate a MLIR program made only of standard dialects into our intermediate program representation. Technically, \PEQC is equipped with a restricted parser for the C language, and exposes an API via specific function calls in the program which are intercepted by the interpreter to model task spawns and the various synchronizations we support. Equipped with this translator, we then rely on existing LLVM MLIR lowering passes for specific dialects (e.g., \MLIRDIAL{affine}) and develop our own MLIR lowering passes for the MLIR-AIR and MLIR-AIE dialects, presented in Secs.~\ref{subsec:mlir-air-convert} onward.

\begin{figure}[h!tb]
\includegraphics[width=12cm]{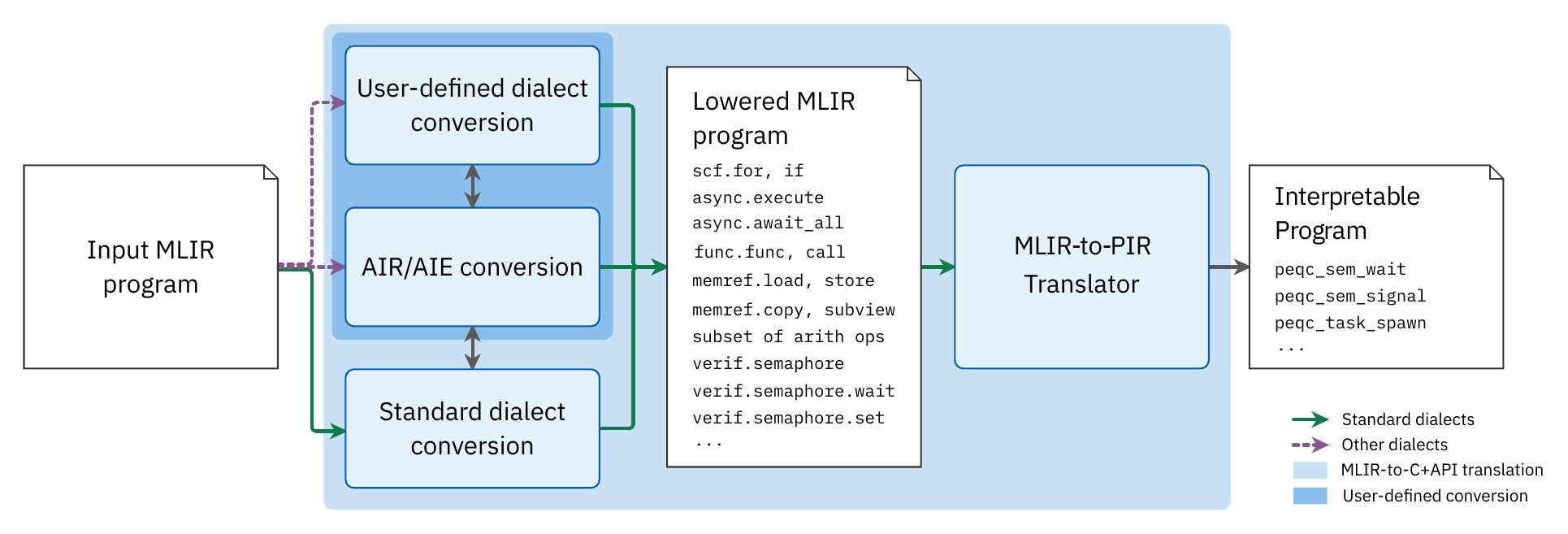}
\vspace{-.5cm}
\caption{\label{fig:flow2}Illustration of conversion and translation to PIR for eventual interpretation}
\vspace{-.35cm}
\end{figure}

\subsection{Translation to \PEQC Intermediate Representation} 
This stage takes as input an MLIR program made only of a subset of \MLIRDIAL{func}, \MLIRDIAL{arith}, \MLIRDIAL{math}, \MLIRDIAL{memref}, \MLIRDIAL{scf} and \MLIRDIAL{async}. It then translates it to a C-syntax AST,
which contains function calls to \PEQC's concurrent API when translating \MLIRDIAL{async}. As discussed below in Sec.~\ref{subsec:mlir-air-convert}, we also provide a dedicated \MLIRDIAL{verif} MLIR dialect in \PEQC, to offer the semaphore types and operations on them (\texttt{wait}/\texttt{set}/\texttt{acquire}/\texttt{release}) and facilitate encoding of concurrent constructs arising in MLIR dialects. The \MLIRDIAL{verif} dialect  translates directly to the \PEQC C API for concurrency.

We note that the simple route of converting an MLIR program to C does not work in general: the \texttt{emitc} dialect and associated tools do not support the programs we target, especially when foreign dialects are used. 

Translation of \MLIRDIAL{arith} and \MLIRDIAL{math} is trivial; we support 
25 operations to date. \MLIRDIAL{scf} and \MLIRDIAL{func} are similarly straightforward. In contrast, \MLIRDIAL{memref} and \MLIRDIAL{async} are more subtle. For \MLIRDIAL{memref}, we emulate operations such as \MLIRDIAL{copy}, \MLIRDIAL{subview}, etc., with a pointwise loop iterating over all elements to be copied. \MLIRDIAL{subview} is treated by creating a temporary variable, and inserting flushes to the original \MLIRDIAL{memref} before the original is next accessed. For \MLIRDIAL{async}, we translate tokens to semaphores and insert the corresponding \MLIRDIAL{set}/\MLIRDIAL{wait} as needed in operations that depend on the tokens, also producing a final semaphore when the \MLIRDIAL{async} operation concludes for subsequent use. We encode \MLIRDIAL{async.group} entries with distinct semaphores, one per member in the group. We do not implement a barrier construct, but instead explicitly create as many point-to-point synchronizations as needed. While this may not be the most efficient, this approach is flexible and allows us to seamlessly model the asynchronous AIR and AIE dialects, OpenMP programs, etc.

\subsection{Leveraging Existing MLIR Lowering Passes}

To maximize the set of MLIR programs we can take as input, we rely on existing lowering passes provided with the standard \texttt{mlir-opt} toolchain. Correctness verification of these lowering passes is outside the scope of this paper;
we focus on verification of programs here.
Specifically, we expect the user to correctly lower programs to a subset of standard dialects from \MLIRDIAL{func}, \MLIRDIAL{arith}, \MLIRDIAL{math}, \MLIRDIAL{memref}, \MLIRDIAL{scf}, and \MLIRDIAL{async}. For example, \texttt{-convert-linalg-to-affine-loops} converts \MLIRDIAL{linalg} operations (e.g., \MLIRDIAL{matmul}) to an expanded (e.g., the 3D loop nest for \MLIRDIAL{matmul}) MLIR \texttt{affine} program. Then \texttt{-lower-affine} lowers the \texttt{affine} dialect to \MLIRDIAL{scf} and \MLIRDIAL{memref}. 

To illustrate, our flow can take as input a \MLIRDIAL{linalg} program from PyTorch, its tiled, optimized, and parallelized variant produced by the MLIR-AIR toolchain, and successfully lower them for eventual translation and verification, as shown in Sec.~\ref{subsec:air-expe}. To achieve this, we designed our own conversion passes to standard dialects for two user-defined dialects: \MLIRDIAL{AIR} and \MLIRDIAL{AIE}. 

\subsection{Lowering of the MLIR-AIR Dialect}
\label{subsec:mlir-air-convert}

The MLIR-AIR dialect captures explicit asynchronous control, hierarchical parallelism, and dataflow between compute and memory, serving as a structured intermediate layer that bridges high-level abstractions and low-level accelerator code generation \cite{wang2025loop}.
While some lowering passes are provided for standard MLIR dialects, user-defined dialects such as AMD's AIR and AIE must also be lowered to standard dialects. We developed these converters for a large subset of AIR and AIE constructs.

Our open-source software contains extensive documentation on operations we support, our approach to translation, and the complete C++ implementation of these converters (built using the LLVM MLIR infrastructure \cite{mlir-web}). We illustrate below three constructs from the AIR dialect. Our implementation supports all AIR operations, albeit with some restrictions on arguments and options supported.

\paragraph*{Conversion of \MLIRDIAL{air.execute}} The \MLIRDIAL{async} dialect enables the description of sets of parallel tasks; tokens produced during task execution enable synchronizations between tasks spawned using \MLIRDIAL{async.execute}. 
We convert \MLIRDIAL{air.execute} to \MLIRDIAL{async.execute}, a process that involves analyzing the \MLIRDIAL{air.token}s used for synchronization in the AIR program and converting them to groups of \MLIRDIAL{async} tasks, and handling value results returned by parallel tasks. This is a form of lowering, as we encode the semantics of the source dialect operation via a new MLIR program fragment made only of standard operations that preserves the original execution semantics. In this case, AIR follows the semantics of an existing standard \MLIRDIAL{async} operation, making conversion straightforward.

\paragraph*{Conversion of \MLIRDIAL{air.herd}} Herds in the AIR dialect represent ``a one or two dimensional array of adjacent AIE cores executing the same code region" \cite{mlir-air-github}. In other words, one can create a parallel task for each AIE core modeled, and execute the same code block on each of them: this amounts to a parallel loop (nest). We therefore convert \MLIRDIAL{air.herd} to \MLIRDIAL{scf.parallel}, a construct capturing the parallelism and constraints of \MLIRDIAL{air.herd}.

\paragraph*{Conversion of \MLIRDIAL{air.channel}} AIR Channels are used to communicate data and synchronize between concurrent entities. In contrast to the above examples, channels do not follow existing constructs in standard dialects. Channels can have blocking read/write operations (e.g., behaving like a blocking FIFO), broadcast to multiple destinations, and optionally perform on-the-fly data layout transformations of the data pushed to or pulled from the channel. 
\vspace{-.25cm}
\begin{lstlisting}[basicstyle=\footnotesize,numbers=none]
// An array of 4 x 4 streaming DMA channels
air.channel @chan_0 [4, 4] {channel_type = "dma_stream"}
// A streaming DMA channel broadcasting to 4 destinations
air.channel @chan_1 [1, 1] {broadcast_shape = [1, 4], channel_type="dma_stream"}
// Receive a 4x4 tile into %dst from channel @chan_0
air.channel.get @chan_0(%dst[%c0, %c0][%c4, %c4][%c1, %c1]) :(memref<16x16xf32>)
// Asynchronous get with dependency on %t1
%t2 = air.channel.get async [%t1] @chan_1(%dst[%c8, %c0][%c4, %c4][%c1, %c1]) : 
      (memref<16x16xf32>)
\end{lstlisting}
To implement conversion of \MLIRDIAL{air.channel} to standard dialects, we developed a conversion using explicit \MLIRDIAL{memref} buffers to store the data (and implement required data layout transformations, by explicitly copying into the proper layout into a new buffer). We developed support for fine-grain synchronizations to model blocking reads/writes if the channel buffer is full. Unfortunately, \MLIRDIAL{async} tokens can only be written once; this prevents modeling semaphores. 
For instance, to model a binary semaphore, one would need to create one token per \emph{access} to the channel, a property immensely difficult to compute statically. To address this issue, we have developed a minimal \MLIRDIAL{verif} dialect, dedicated to handling semaphores (\MLIRDIAL{wait}, \MLIRDIAL{set}, \MLIRDIAL{acquire}, and \MLIRDIAL{release}) as well as the type \MLIRDIAL{verif.semaphore}. The \MLIRDIAL{verif} dialect has direct conversion to the \PEQC API. Consequently, \MLIRDIAL{air.channel}s are converted to blocks of code containing: \MLIRDIAL{memref}s to store data, loops to copy data into a different layout, one semaphore per channel, and \MLIRDIAL{wait}s/\MLIRDIAL{set}s around read/write operations.

\subsection{Lowering of the MLIR-AIE Dialect}
\label{subsec:mlir-aie-convert}

While MLIR-AIR focuses on higher-level constructs useful for representing computations on spatial architectures generally, MLIR-AIE is focused primarily on programming constructs that map cleanly to Neural Processing Units \cite{mlir-aie-github}.
We have implemented support for the following \MLIRDIAL{aie} operations: \MLIRDIAL{buffer}, 
\MLIRDIAL{core}, \MLIRDIAL{device}, \MLIRDIAL{dma\_bd}, \MLIRDIAL{flow}, \MLIRDIAL{lock}, \MLIRDIAL{mem}, \MLIRDIAL{memtile\_dma}, \MLIRDIAL{next\_bd}, \MLIRDIAL{objectfifo},  \MLIRDIAL{objectfifo.acquire}, \MLIRDIAL{objectfifo.createObjectFifo}, \MLIRDIAL{objectfifo.link}, \MLIRDIAL{objectfifo.release}, \MLIRDIAL{objectfifo.subview.access}, \MLIRDIAL{tile}, \MLIRDIAL{use\_lock}, \MLIRDIAL{npu.dma\_memcpy\_nd}, \MLIRDIAL{runtime\_sequence}. This is not the full AIE dialect; specifically, operations related to low-level routing (e.g., \MLIRDIAL{packet\_flow}, \MLIRDIAL{switchbox}) are not currently implemented. 

Specifically, concurrent MLIR-AIE constructs such as \MLIRDIAL{aie.core} and \MLIRDIAL{aie.mem} are translated to parallel tasks in the \MLIRDIAL{async} dialect.  Lock operations, such as \MLIRDIAL{lock} and \MLIRDIAL{use\_lock}, represent counting semaphores and are lowered into the \MLIRDIAL{verif} dialect.  More complex communication constructs, such as \MLIRDIAL{aie.objectFifo}, are lowered into a combination of buffer allocations, data copy operations, and semaphores.  The main complexity here is representing complex data layout transformations that can be implemented as part of \MLIRDIAL{aie.objectFifo} accesses.  Most other constructs lower into standard MLIR dialects representing sequential behavior, although some concurrent subtleties exist as discussed later in Sec.~\ref{sec:aie-expe}.

%% file: experiments.tex
We now present experimental results on correctness verification of a variety of MLIR-based optimization flows. We present detailed case studies for AIR in Sec.~\ref{sec:air-expe} and AIE in Sec.~\ref{sec:aie-expe}, demonstrating the ability of \PEQC to verify optimized, high-performance implementations for the AI Engine, and to uncover subtle potential concurrency bugs. We then present extensive evaluation of \texttt{mlir-opt} on the PolyBench/C benchmarking suite \cite{polybench}, demonstrating \PEQC's ability to quickly find numerous known limitations but also unknown bugs in the tested toolchain in Sec.~\ref{sec:polybenchexpe}.

We have implemented all our tools as open-source, BSD-licensed software \cite{peqc-mlir-github}. \PEQC was evaluated using a single CPU core from an Intel Xeon E3-1240v6 server running at 3.70GHz, with 32GB of memory.

\subsection{Verification of MLIR-AIR}
\label{sec:air-expe}
\label{subsec:air-expe}

To demonstrate the ability of \PEQC to verify pairs of programs using possibly different MLIR dialects prior to conversion, Table.~\ref{tab:airexpe1} displays statistics for the evaluated program variants. All variants are produced as intermediate steps in the MLIR-AIR compilation flow and eventually are lowered to MLIR-AIE (further steps are evaluated in Sec.~\ref{subsec:aie-expe}). Fig.~\ref{fig:flowdiagramexpeairaie} displays the flow of this correctness verification approach.
\begin{figure}[h!tb]
\vspace{-.3cm}
\includegraphics[width=11cm]{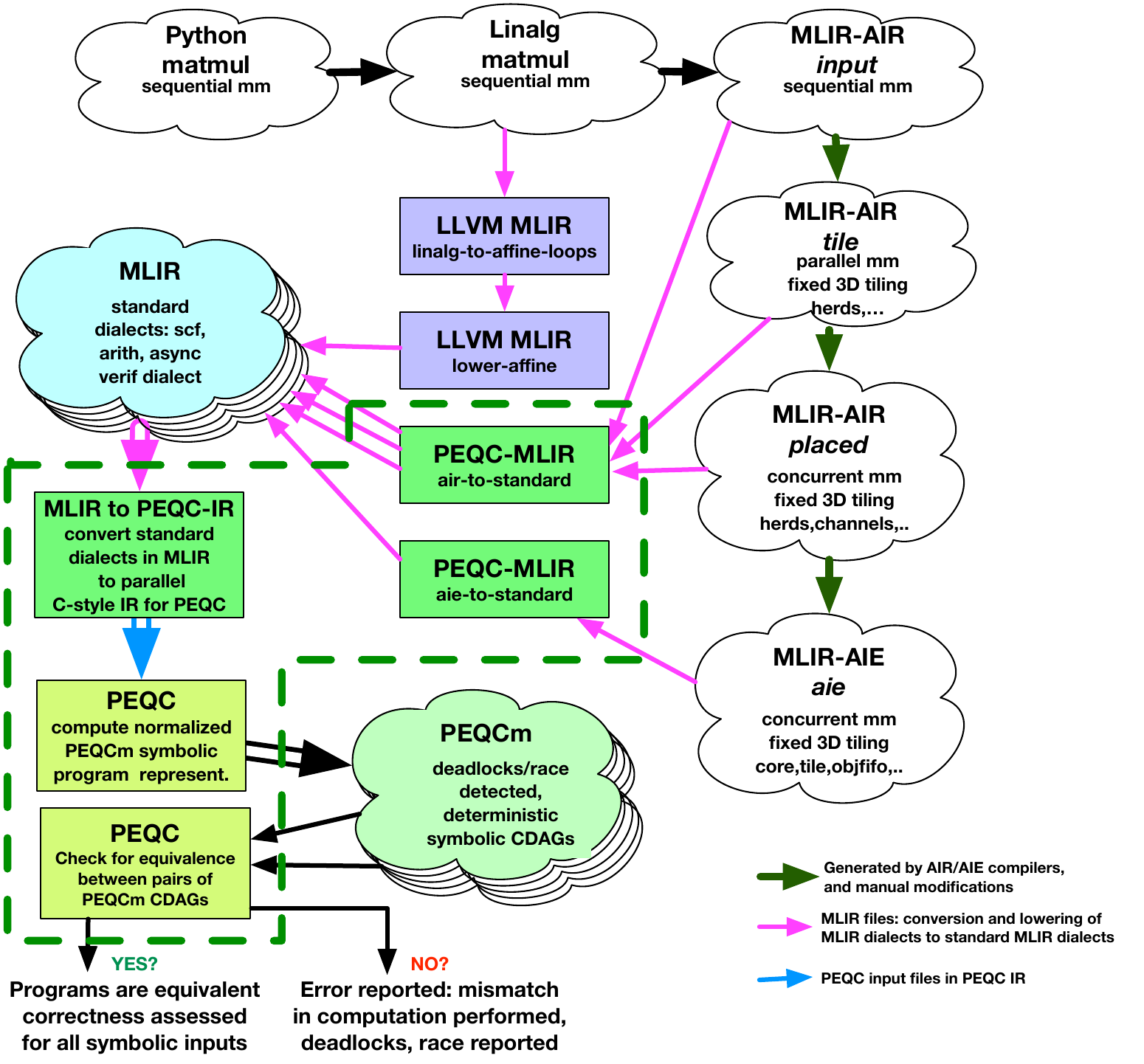}
\caption{\label{fig:flowdiagramexpeairaie}Flow of AIR/AIE correctness checking with \PEQC.}
\vspace{-.3cm}
\end{figure}

The \MLIRDIAL{linalg} program is built from a Python specification, and represents a sequential dense matrix-multiplication kernel, for a matrix (block) of size $32x32$ to be run on the AI Engine.
\textsf{air.input}, the reference input program to be compiled by the MLIR-AIR toochain into a valid MLIR-AIE implementation on the AI Engine, was built by applying lowering of \MLIRDIAL{linalg} to \MLIRDIAL{affine} then to \MLIRDIAL{scf}.
\textsf{air.tiled} is a hierarchically tiled transformation of \textsf{air.input}, specified with a user-defined transformation script using the MLIR Transform dialect. It introduces AIR operations, such as \MLIRDIAL{segments}, DMA copies, etc.
\textsf{air.placed} is optimized and lowered from \textsf{air.tiled}, and introduces \textsf{air.channels} for communications.

These programs are verified pairwise as illustrated in Fig.~\ref{fig:flowdiagramexpeairaie}. First, they are converted to equivalent standard MLIR dialects (using existing LLVM MLIR lowering passes, or \PEQC lowering passes we have developed), then translated to \PEQC's input Parallel Intermediate Representation (PIR) as per Sec.~\ref{sec:mlirsupport}. These PIR files are then checked pairwise for equivalence, following Sec.~\ref{sec:formal}. 

\begin{table}[h!tb]
\vspace{-.3cm}
\caption{\label{tab:airexpe1}Statistics on intermediate files produced during AIR compilation flow.}
\vspace{-.4cm}
{\small
\begin{tabular}{l|r|r|r|r|r|r|r|r|r}
bench & S-LoC & C-LoC & I-LoC & nb-s & nb-conc & nb-sync & hb-rat & $t_{int}$ & $t_{eq}$ \\ \hline

\textsf{linalg} & 8 & 17 & 56 & 201870 & 0 & 0 & 0 & 1.61s & N/A\\
\textsf{air.input} & 17 & 51 & 56 & 541420 & 0 & 0 & 0 & 3.39s& 0.04s\\
\textsf{air.tiled} & 102 & 133& 160 & 743594 & 25 & 25 &0.0003 & 5.56s & 0.04s\\
\textsf{air.placed} & 190 & 902 & 1030 & 1320799 & 440 & 77815 & 0.39 & 9.01s & 0.04s\\ \hline
\end{tabular}
}
\vspace{-0.2cm}
\end{table}
We report in Table.~\ref{tab:airexpe1} metrics about the program source (S-LoC: lines of code of original MLIR file; C-LoC: for the converted program using only standard dialects; I-LoC: for the final PIR file checked for correctness), and evaluation metrics (nb-s: total number of statements interpreted; nb-conc: total number of tasks created; nb-sync: total number of synchronizations). hb-rat is the ratio between the number of nodes in our representation of the graph $G$ divided by nb-s. 
Finally, $t_{int}$ is the total time to interpret, that is to produce the final memory state for the program, and $t_{eq}$ is the total time to check equivalence between two memory states, with the one of the previous file. In this set of experiments, all files were proved equivalent, leading to the maximum equivalence time (the process ran to completion).

Interpretation time follows roughly the number of statements interpreted, while equivalence time remains constant since the memory states for non-local variables are identical in these experiments (the programs are equivalent, for symbolic input matrices but a fixed problem size).

\subsection{Verification of AIE}
\label{sec:aie-expe}
\label{subsec:aie-expe}

\PEQC supports a rich subset of AIE operations, as described in Sec.~\ref{subsec:mlir-aie-convert}.
We now illustrate an important case study, demonstrating the merits of \PEQC in catching subtle concurrency bugs, as well as the ability to verify equivalent a low-level ``hardware-friendly'' optimized MLIR program with the original, sequential \MLIRDIAL{linalg} program. In this flow, the final AIR program with channels and placement is lowered and optimized to an AIE representation with \texttt{air-opt} and \texttt{aie-opt}.

The final AIE file is 753 lines of code, leading to about 2.2k LoCs for the final PIR file obtained after conversion and translation. In this case, \PEQC reported a possible race. Fig.~\ref{fig:bugar1} is an excerpt from the matrix-multiply AIE program, but zooming in on a chain of acquire-release for a task to execute on a PE (a memory tile). \texttt{aie.dma\_start(S2MM, 0, \^{}bb1, \^{}bb5)} starts a DMA communication, in the S2MM direction, and both basic blocks \texttt{\^{}bb1} and \texttt{\^{}bb5} are started in parallel.

\begin{figure}[h!tb]
\begin{lstlisting}[basicstyle=\footnotesize,language=c++,escapeinside={$}{$}]
%memtile_dma_2_1 = aie.memtile_dma(%mem_tile_2_1) {
      %0 = aie.dma_start(S2MM, 0, ^bb1, ^bb5)
    ^bb1:  // 2 preds: ^bb0, ^bb1
      aie.use_lock(%lock_2_1, AcquireGreaterEqual, 2)
      aie.dma_bd(%buf20 : memref<64x16xi32, 1 : i32>, 0, 1024)
      {task_id = 0 : i32}
      aie.use_lock(%lock_2_1_0, Release, 2)
      aie.next_bd ^bb1
    ^bb2:  // pred: ^bb3
      aie.end
    ^bb3:  // pred: ^bb5
      %1 = aie.dma_start(MM2S, 0, ^bb4, ^bb2)
    ^bb4:  // 2 preds: ^bb3, ^bb4
      aie.use_lock(%lock_2_1_0, AcquireGreaterEqual, 1)
      aie.dma_bd(%buf20 : memref<64x16xi32, 1 : i32>, 0, 512, 
      [<size = 64, stride = 16>, <size = 8, stride = 1>]) 
      {task_id = 0 : i32}
      aie.use_lock(%lock_2_1, Release, 1)
      aie.next_bd ^bb4
    ^bb5:  // pred: ^bb0
      %2 = aie.dma_start(MM2S, 1, ^bb6, ^bb3)
    ^bb6:  // 2 preds: ^bb5, ^bb6
      aie.use_lock(%lock_2_1_0, AcquireGreaterEqual, 1)
      aie.dma_bd(%buf20 : memref<64x16xi32, 1 : i32>, 8, 512, 
      [<size = 64, stride = 16>, <size = 8, stride = 1>])
      {task_id = 0 : i32}
      aie.use_lock(%lock_2_1, Release, 1)
      aie.next_bd ^bb6
}
\end{lstlisting}
\caption{\label{fig:bugar1}Excerpt of \texttt{acquire}/\texttt{release} nondeterminism.}
\end{figure}

This code is nondeterministic: in the typical case, both \texttt{bb4} and \texttt{bb6} acquire \texttt{\%lock\_2\_1\_0} once and execute exactly once.
In the case of significant stalling on the hardware implementing \texttt{bb6}, however, \texttt{bb4} may acquire the lock and execute twice before \texttt{bb6} is able to acquire the lock once.
Note that the \texttt{aie.dma\_bd} operations in \texttt{bb4} and \texttt{bb6} do not implement the same movement due to different arguments. \PEQC reports this nondeterminism, analyzing the file in less than 5s (early exit on first error, otherwise 31s for a corrected, equivalent version). This is a case where testing alone did not discover the concurrency error, which is unlikely to be triggered except in unusual scenarios involving significant stalling.  The implemented fix \cite{mlir-aie-github} uses additional locks to ensure no two basic blocks can access the same buffer.

\subsection{Verification of \texttt{mlir-opt} and \texttt{polymer-opt}}
\label{sec:polybenchexpe}

\input{expe-polybench.tex}

%% file: expe-polybench.tex
We conducted an extensive evaluation of several tools: \texttt{mlir-opt} from LLVM 19.0.0git; \texttt{cgeist}, and \texttt{polygeist-opt} and \texttt{polymer-opt} from LLVM 18.0.0git, the latest versions available from \cite{cgeist-web}. 

\texttt{mlir-opt} is the base MLIR compiler provided with LLVM, and offers numerous conversion and optimization passes, as illustrated below. \texttt{cgeist} takes a C program as input, and raises it to the \MLIRDIAL{affine} dialect when possible. \texttt{polymer-opt} is a companion to \texttt{cgeist}, taking \MLIRDIAL{affine} MLIR programs as input and calling the Pluto polyhedral optimizer \cite{uday08pldi} on them, to implement automatic loop transformations for tiling, fusion, and parallelization. We chose these three tools as they exhibit different levels of maturity: \texttt{mlir-opt} is a community-maintained LLVM package, \texttt{cgeist} is a separate project \cite{cgeist-web} developed by MLIR authors, and Pluto is a research compiler. 

For benchmarks, as we can exercise \texttt{cgeist} and Pluto with only affine programs, we chose PolyBench/C version 4.2.1 \cite{polybench-web}. Detailed statistics for these benchmarks are available in the Appendix. PolyBench/C covers key computation patterns in linear algebra, image processing, and machine learning.

We systematically raise the PolyBench C benchmark kernel function to MLIR with \texttt{cgeist} and then apply \texttt{polymer-opt} and/or \texttt{mlir-opt} optimizations. Every resulting MLIR file is converted and translated to a \PEQC-ready input, and verified for equivalence. 

In total, we consider 22 different optimization paths for 30 benchmarks, leading to 580 variants actually generated. For some cases, the toolchain (\texttt{mlir-opt} and/or \texttt{polymer-opt}) fails to produce a program. Running the entire set on the PolyBench mini dataset takes less than 30 minutes total (about 15 min. to generate transformed files with MLIR, and about 10 minutes to convert and check them for equivalence). As shown below, this still allows \PEQC to uncover a variety of bugs and possible misuse of some tools (which are not documented), further justifying the inclusion of such a fast verification scheme in compiler testing procedures. We also evaluated with the medium dataset size, obtaining identical bugs/correctness results.

\def \mlirfuse {FS}
\def \mlirlicm {LICM}
\def \mlirnorm {NM}
\def \mlircanon {CAN}
\def \mlirtile {TL}
\def \mlirunroll {UR}
\def \mliruaj {UAJ}
\def \mlirpar {PAR}
\def \mlirpluto {Pluto}
\def \FuseAndTile {FuseTile}
\def \mlirAll {All}

\paragraph*{List of optimizations considered} Table~\ref{fig:optimtable1} displays the various optimization paths we evaluate. 
They are not meant to capture all possible optimizations, but rather classical ones that some may wish to apply especially to dense linear algebra programs: tiling, fusion, unroll-and-jam, parallelization; and, of course, Pluto-based polyhedral tiling+fusion, as a first transformation before these \texttt{mlir-opt} optimizations are invoked. 

\begin{table}[h!tb]
\vspace{-.0cm}
\caption{\label{fig:optimtable1}List of individual optimizations considered.}
\begin{minipage}{0.99\textwidth}
\centering
    \small
    \begin{tabular}{l||l}
    Acronym & command \\ \hline
        \mlirfuse & \texttt{mlir-opt --affine-loop-fusion} \\ 
        \mlirlicm & \texttt{mlir-opt --affine-loop-invariant-code-motion} \\ 
        \mlirnorm & \texttt{mlir-opt --affine-loop-normalize} \\ 
        \mlircanon & \texttt{mlir-opt --canonicalize} \\ 
        \mlirtile & \texttt{mlir-opt --affine-loop-tile="tile-size=32"} \\
        \mlirunroll & \texttt{mlir-opt --affine-loop-unroll} \\ 
        \mliruaj & \texttt{mlir-opt --affine-loop-unroll-jam} \\ 
        \mlirpar & \texttt{mlir-opt --affine-parallelize="max-nested=1"} \\ 
        \mlirpluto & \texttt{polymer-opt -reg2mem -extract-scop-stmt -pluto-opt} \\ 
    \end{tabular}
\end{minipage}
\end{table}

\paragraph*{Verification Results for \texttt{mlir-opt}} We now present our results, first on verifying only \texttt{mlir-opt}, 
without \texttt{polymer-opt} as a first transformation, in Table~\ref{table:verifnopluto}. That is, \texttt{mlir-opt} is applied on the original program as extracted by \texttt{cgeist}. This table displays the benchmarks \emph{that fail to be proved equivalent}; if no benchmark fails, N/A is displayed. 

\begin{table}[h!tb]
\vspace{-0.0cm}
\begin{minipage}{0.99\textwidth}
    \centering
    \caption{\label{table:verifnopluto}List of failed benchmarks for \texttt{mlir-opt} without Pluto/\texttt{polymer-opt} optimizations. Note ADI and durbin are not displayed: the code produced by Cgeist fails to be proved equivalent to the original program, meaning they fail in all cases.}
    \vspace{-.25cm}
{   \centering \small
    \begin{tabular}{l|l}
    Optimization & Failing benchmarks \\ \hline
        \mlirfuse & correlation, covariance, deriche,  \\ 
        & doitgen \\
        \mlirlicm & N/A \\ 
        \mlirnorm & N/A \\ 
        \mlirtile & floyd-warshall, seidel-2d \\ 
        \mlirunroll & N/A \\ 
        \mliruaj & floyd-warshall \\ 
        \mlirpar & N/A \\

        \mlirnorm, \mlircanon, \mlirtile, \mliruaj, \mlirtile & floyd-warshall, seidel-2d \\ 
        \mlirnorm, \mlircanon, \mlirfuse, \mlirtile & correlation, covariance, deriche, \\
        & doitgen, floyd-warshall, seidel-2d \\
        \mlirunroll, \mliruaj, \mlirtile, \mlirlicm, \mlirfuse, \mlirnorm & floyd-warshall, seidel-2d \\ 
        \mlirunroll, \mliruaj, \mlirtile, \mlirpar, \mlirlicm, \mlirfuse, \mlirnorm & floyd-warshall, seidel-2d        \\
    \end{tabular}
    }
\end{minipage}
\end{table}

We observe numerous errors, which we investigated and confirmed manually. Several of these errors have in fact also been recently reported with HEC \cite{yin25}.
For example, fusion on imperfectly nested loops such  as correlation, doitgen, etc. is incorrectly implemented: inner/no-outer loops with non-forward dependences have been fused for these benchmarks. More interestingly, for Floyd-Warshall any tiling or unroll-and-jam is incorrect: \texttt{mlir-opt} tiled all 3 loops, while the $k$ loop is not tilable as-is on FW. On the other hand, the tiling produced is valid for all other benchmarks except Seidel-2d (which suffers a similar issue). We particularly note that tiling is \emph{correct} for other stencils: for the Jacobi stencils for example only the time loop is strip-mined. But for Seidel-2d, all 3 loops are tiled, without skewing, which is incorrect. 

In practice, a question remains open: are these optimization flags in \texttt{mlir-opt} supposed to detect correctly whether they can be applied on the input MLIR affine program? Their implementation suggests such analysis is not implemented, yet there are numerous cases where applying these optimizations lead to a correctly transformed MLIR program. \PEQC can help systematically uncover such potential misuse in a mechanical fashion, by testing systematically e.g. on PolyBench as we conducted above. However, as shown below, it can also uncover critical bugs in tools that perform such analysis and transformation selection automatically, like Pluto/Polymer.

\paragraph*{Verification results for Pluto and \texttt{mlir-opt}}
We now report the same evaluation, but the input MLIR programs have first been optimized by \texttt{polymer-opt}, invoking the polyhedral Pluto optimizer with default tile size of 32. Pluto raises the program to polyhedral representation, transforms it, and regenerates an AST for the fused/tiled/skewed program, itself converted back to MLIR \texttt{affine} by \texttt{polymer-opt}. Note the \texttt{mlir-opt}
 optimizations considered operate best on an affine MLIR program without function calls, so after Pluto but prior to \texttt{mlir-opt} optimizations we systematically invoke \texttt{mlir-opt} \texttt{--inline}, which can itself sometimes fail. In particular, ADI, ludcmp, nussinov and symm fail at inlining and are not displayed below.

Table~\ref{table:plutoLA} displays the list of failed benchmarks that involve linear-algebra-style computations. We observe systematic errors on covariance, lu and gramschmidt in particular: the produced code performs out-of-bound array accesses. We confirmed the bug by simply displaying the array locations accessed (a printf) in the C code printed by Pluto during debugging: it does produce negative indices. We remain unsure of the exact reason of these bugs: the program produced by \texttt{cgeist} is equivalent to the original program, but Pluto-PolyLib itself seems to generate incorrect code. This could be due to an incorrect polyhedral representation extracted, or a problem in \texttt{polymer-opt}.
\vspace{-.2cm}
\begin{table}[h!tb]
    \centering
    \caption{\label{table:plutoLA}List of failed benchmarks for Pluto/\texttt{polymer-opt} followed by \texttt{mlir-opt} optimizations, for the linear algebra benchmarks (polybench/linear-algebra and polybench/datamining benchmarks). NG indicates the file was Not Generated by \texttt{mlir-opt}.}
    \vspace{-.25cm}
{    \small
    \begin{tabular}{ m{16em} | m{24em} }
    Optimizations & List of failed benchmarks \\ \hline
        \mlirpluto, \mlirfuse & correlation, covariance, gramschmidt, lu, \\
        \mlirpluto, \mlirlicm & covariance, gramschmidt, lu \\ 
        \mlirpluto, \mlirnorm & covariance, gramschmidt, lu  \\ 
        \mlirpluto, \mlirtile & covariance, gramschmidt, lu(NG) \\
        \mlirpluto, \mlirunroll & covariance, gramschmidt, lu \\
        \mlirpluto, \mliruaj & covariance, gramschmidt, lu, \\ 
        \mlirpluto, \mlirpar & covariance, doitgen, gramschmidt, trmm \\ 
        \mlirpluto, \mlirnorm, \mlircanon, \mlirtile, \mliruaj, \mlirtile & correlation, covariance, gramschmidt, lu \\ 
        \mlirpluto, \mlirnorm, \mlircanon, \mlirfuse, \mlirtile & correlation, covariance, gramschmidt, lu \\ 
        \mlirpluto, \mlirunroll, \mliruaj, \mlirtile, \mlirlicm, \mlirfuse, \mlirnorm & covariance, gramschmidt, lu(NG) \\ 
        \mlirpluto, \mlirunroll, \mliruaj, \mlirtile, \mlirpar, \mlirlicm, \mlirfuse, \mlirnorm & covariance, doitgen, gramschmidt, lu(NG) \\    \end{tabular}
    }
\label{tab:fail-pluto-la}
\vspace{-.25cm}
\end{table}
For stencil computations (jacobi-1d, jacobi-2d, heat-3d, fdtd-2d, seidel-2d), we omit the table which is summarized instead as observing \emph{systematic} failures for fdtd-2d, jacobi-2d and seidel-2d. All other stencil benchmarks pass. For these 3 failing stencils, they suffer a similar issue as above: the code produced after Pluto's skewing is incorrect, and performs an out-of-bound (-1) array access along the j dimension, which we confirmed by printing the actual value of the accessed locations in the program printed by Pluto. We are unsure of the location of this problem, but remark that we exercise a particular case with the MINI dataset: often the problem size is smaller than the tile size, and we inline loop bounds while Pluto is designed for parametric loop bounds. We suspect a possible issue with the \texttt{ceild}/\texttt{floord} implementation in this case, but further investigation is needed. Note that these benchmarks also fail similarly with larger problem sizes (MEDIUM) which far exceed the tile sizes.

We also remark that, of course, not all optimizations led to transforming all programs: \texttt{mlir-opt} does bypass applying these optimizations in numerous cases (supposedly where they are not applicable), but still applies them incorrectly in several cases, as reported above. 

\paragraph*{Verification throughput experiments} 
We conclude this analysis by displaying in Fig.~\ref{fig:throughputplot} the throughput, in number of statements interpreted per second, for all experiments conducted above, distinguishing MINI and MEDIUM dataset sizes. 
For the MINI dataset, times are so low that the time to analyze the input program (compute symbol tables, normalize) has a measurable effect, but is marginalized when time is dominated by actual interpretation as in the medium dataset.  For \texttt{mlir-opt}, statements are typically 3-address and therefore faster to interpret, but there are often 10x more statements total to interpret.

\begin{figure}[h!tb]
\vspace{-.0cm}
  \begin{center}
    \includegraphics[width=0.69\textwidth]{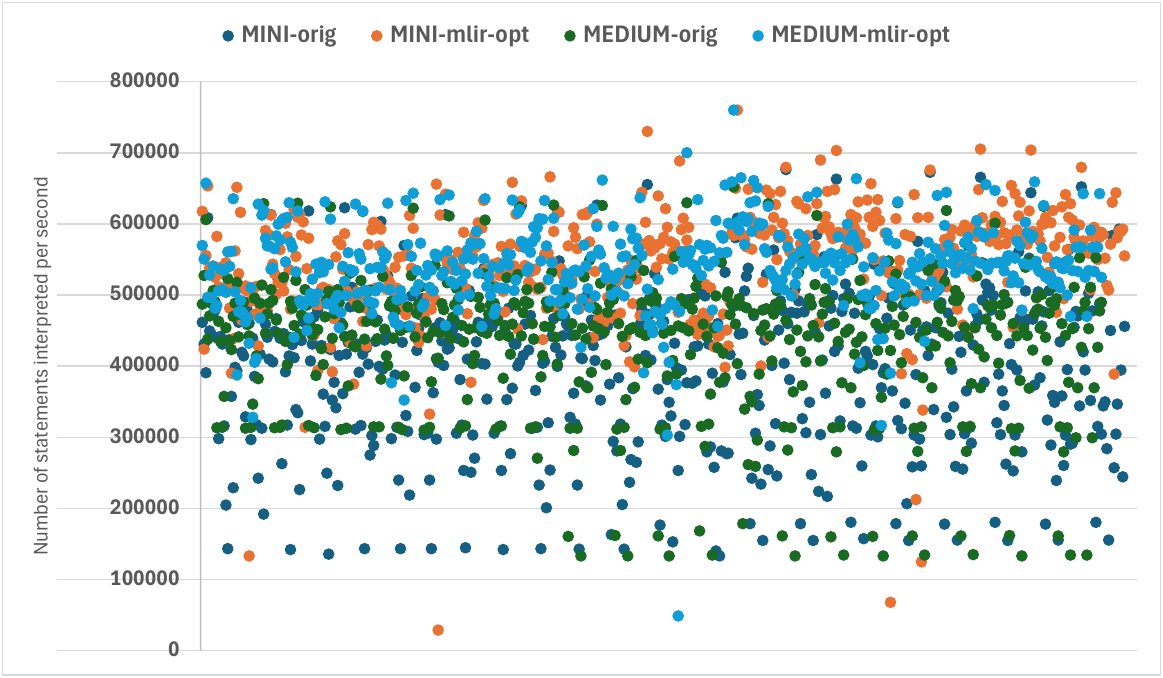}
  \end{center}
  \caption{\label{fig:throughputplot}Throughput of the verifier: interpretation time only, over 2160 program variants built with \texttt{mlir-opt}. Equivalence time is negligible vs. interpretation time.}
  \vspace{-.3cm}
\end{figure}

\subsection{Scaling to Larger Problem Sizes}

We conclude our experimental evaluation by briefly outlining ongoing work we are conducting to improve further scaling to larger problem sizes. 
Scaling to larger problem sizes  is mostly limited by memory consumption, reaching potentially tens of GB for e.g. a 256x256 triple matrix product \cite{fpga24}. In practice, the CDAGs are stored explicitly, with one node per operation executed. We are currently developing techniques to automatically synthesize \emph{code} which when executed implements exactly the CDAG: this is a form of (affine) folding from the trace \cite{rodriguez2016trace,pldi2019augustine}. Although we have only preliminary results for only sequential program verification, our prototype already demonstrates very strong potential to increase scaling to larger problem sizes.

For example, for matrix-products, we remove this memory limitation by automatically folding CDAGs during their construction into a sort of "chain of recurrence" representation, limiting memory usage to spatial dimensions (that is typically the size of the input data) but becoming independent of the reduction size, providing an $O(N)$ improvement in space complexity for GEMM-style kernels for example. 

In particular, our preliminary CDAG compaction approach excels for programs with \emph{affine evolution} of the computations formed to produce the output values (this even if the actual schedule of operations implemented is itself non-affine), as is typical in e.g. dense linear algebra and inference of deep learning models, where sequences of matrix/tensor manipulations and multiplications are dominant.

\begin{figure}[h!tb]
  \begin{center}
    \includegraphics[width=0.77\textwidth]{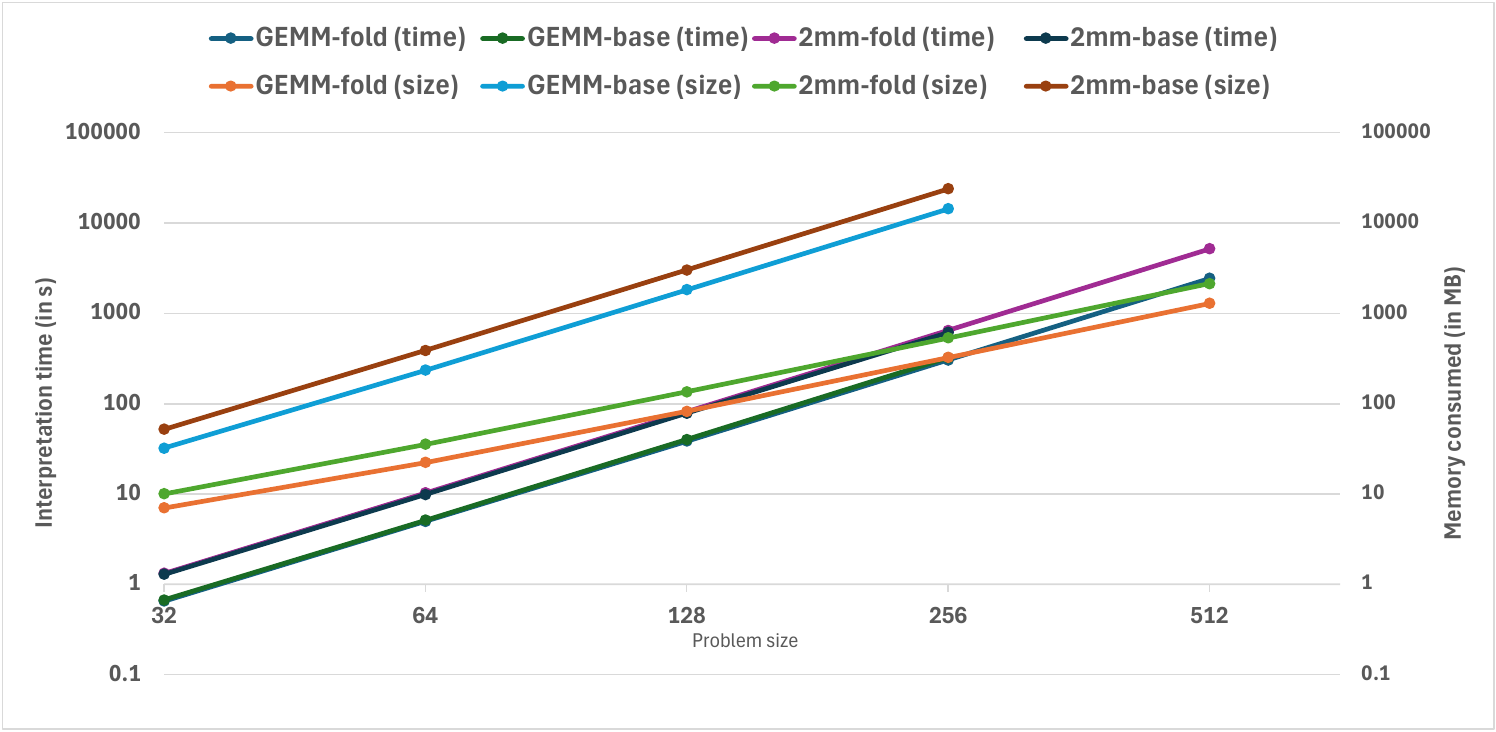}
  \end{center}
  \caption{\label{fig:scaling}Time and maximal memory usage for gemm and 2mm, end-to-end verification. Fold uses affine folding, base uses normal CDAGs}
  \vspace{-.3cm}
\end{figure}  

We display in Fig.~\ref{fig:scaling} the scaling achieved to verify two versions (original and optimized variants, not buggy) of two benchmarks: gemm and 2mm from PolyBench, as they particularly exercise the need to compress reductions. Visibly, the memory consumption does not grow linearly with instructions executed ($O(N^3)$) anymore when folding is activated. It leaves time to verification as the only "limit", and can achieve orders of magnitude reduction in memory usage, as seen on the green and orange lines, which increase with lower slopes than the base version ($O(N^2)$) using folding.

%% file: related.tex
While we are unaware of any other approach with the same coverage as experimented in Sec.~\ref{sec:experiments}, especially with the range of transformations and concurrency/parallelism we support, numerous prior works target related problems (and may handle some of the simpler benchmarks in our 
\texttt{mlir-opt} evaluation), and their coverage may intersect with ours.

UC-KLEE \cite{ramos2015under} targets equivalence of programs with complex control-flow, via symbolic execution. However, it often fails to scale even for somewhat simple cases we support. KLEE \cite{klee-web} was extended to support floating-point \cite{liew2017floating}, but does not scale to the problem sizes we target, by orders of magnitude.

Another fundamental direction to prove correctness of transformations is translation validation (TV) \cite{necula2000translation}. Alive2 provides bounded TV for LLVM \cite{lopes2021alive2}, and it excels at catching a wide class of intricate bugs related e.g., to incorrect instructions emitted. However, it is not able to prove equivalence for the scope of programs we target. Other approaches for TV tend to restrict the scope of programs and optimizations supported, yet deliver significant bug detection capability, including for MLIR e.g., \cite{bang2022smt,wang2024systematic}. TV has also been deployed for specialized languages like Halide \cite{clement2022end}, and for HLS \cite{chouksey2020hls,karfa2006hls}. 
Recent work presented frameworks to develop verifiable compilers, e.g., \cite{rosu10,eldridge21} complementing existing certified compilers, e.g. \cite{leroy2016compcert,vsevvcik2013compcerttso,herklotz2021formal}.

Numerous works target problems related to the correctness of MLIR programs, typically by defining formally the behavior of transformations or at least equivalence between them, e.g., \cite{bhat24, fehr25, yin25, wang24a} which all target the verification of MLIR components, albeit for a different scope than \PEQC. 
In contrast to these techniques, we typically trade off requiring additional information to limit the program scope (such as the problem size), but broaden significantly the set of program transformations handled therein. No prior work to our knowledge addresses the correctness verification of MLIR-AIR and MLIR-AIE toolchains.

More recently, equality saturation has reached a level of performance allowing impactful transformations to be modeled \cite{tate2009equality,wang2020spores,willsey2021egg,zhang2023better,flatt, merckx25}; and is deployed for MLIR verification \cite{yin25}.  However, these techniques are far from covering the rewrites necessary to model the class of transformations we support, especially regarding concurrency, synchronizations, and parallelization. But approaches such as \texttt{egglog} may be used to normalize our CDAGs further in lieu of of our rewrite rule system~\cite{zhang23}.

Other approaches to checking or guaranteeing the correctness of (parallel) programs include model checking e.g.,  \cite{mercer2015model,vechev2010automatic,siegel15,jakobs21a,golovin25,kokologiannakis19}, as well as static and dynamic race detection approaches \cite{gu18, swain20, atzeni16} which often also build a representation of the happens-before relation to detect possible races \cite{gu18}. SMT-based techniques such as CIVL \cite{siegel15,wu2022verifying,kokologiannakis19} have the potential ability to reason on a wider class of programs
(e.g., not requiring fixed problem size), however care must be taken that in practice, most often strong assumptions are made on the range of some variables (such as fixing the problem size to one or a few values), limiting otherwise scaling/coverage for complex transformed programs such as those we support.

PolyCheck is a static+dynamic analysis approach to  prove equivalence between a sequential affine program and its (possibly non-affine) transformation \cite{polycheck}. 
It supports arbitrary iteration reordering transformations, and its space complexity is linear in the data size touched by the program, not the number of operations performed. 
In contrast, our framework does not require any of the programs to be affine, and proves equivalence for concurrent programs. 
ISA \cite{verdoolaege2012equivalence} targets the proof of equivalence between a pair of affine programs, including parametric loop bounds, under a limited yet powerful set of transformations. However, both remain highly limited overall in code transformation coverage and rely on polyhedral static analyses \cite{polycheck}. 
Numerous other approaches have been developed to prove equivalence e.g. of expression trees or limited forms of programs \cite{shashidhar2005verification,karfa2013verification} however their coverage falls short of supporting the complex programs we target. 

Prior works have reasoned formally about determinism for parallel languages, such as for X10~\cite{lee10}, Concurrent Collections \cite{budimlic10}, and Habanero-C with promises \cite{jin23}. Our formalism in Sec.~\ref{sec:formal} took initial inspiration from these works, but we profoundly broadened the scope of concurrency and synchronization styles (e.g., binary and counting semaphores) that are proved to be deterministic, and also enable reasoning on symbolic variables.

%% file: conclusion.tex
Our proposed approach to correctness verification of a class of MLIR programs was presented, formalized, and evaluated extensively. It provides formal guarantees of correctness for a highly useful class of programs.
The lowering passes implemented in MLIR and by users would also benefit from correctness verification. To address this issue and provide verified converters in the long term, we argue that designers should formalize the semantics of their dialect ideally in terms of operations from standard dialects, equipped with clear semantics for standard dialects (which itself still needs to be developed). We hope this paper encourages the community to work on formalizing the semantics of dialects and their lowering to standard dialects, and leverage \PEQC's formalization for verification we presented in this paper.

%% file: appendix.tex
\section{Appendix}

\subsection{Implementation Details: Debugging Features}

\input{userfeedback.tex}

\subsection{Per-benchmark Evaluation Metrics}

\paragraph*{Basic statistics on benchmarks}
We display details on the benchmarks evaluated.
Table~\ref{table:allpolybench} displays the list of all benchmarks considered, as well as some statistics on the number of statements interpreted, number of nodes in the CDAGs produced, memory use and time to interpret the original programs, for both the MINI and MEDIUM dataset sizes of PolyBench/C. Interpretation time of transformed variants follows similar trends.
However, often more statements are produced via transformations, so statement count and interpretation time increases accordingly, sometimes by up to 3x. 

\begin{table}[!ht]
    \centering
    \caption{\label{table:allpolybench}List of all 30 PolyBench/C benchmarks, and statistics on interpreting the original input benchmarks (number of statements interpreted, number of CDAG nodes eventually created, time and space to interpret) for both the MINI and MEDIUM datasets}
    \vspace{-.25cm}
{    \scriptsize    
    \begin{tabular}{l|l|l|l|l|l|l|l|l|l}
    benchmark & description & \multicolumn{2}{c|}{\#stmts} & \multicolumn{2}{c|}{\#cdag nodes} &
    \multicolumn{2}{c|}{Mem. used (B)} & \multicolumn{2}{c}{Time (s)} \\ 
     & & MINI & MED & MINI & MED &MINI & MED &MIN & MED \\ \hline   
        2mm & 2 Matrix Mult. $(\alpha A B C + \beta D)$ & 43.2k & 44.5M & 102.4k & 110.2M & 6.15M & 6.48G & 0.14 & 139 \\ \hline
        3mm & 3 Matrix Mult. $((A B)(C D))$ & 70.2k & 68.9M & 134.9k & 137M & 8.12M & 8.08G & 0.23 & 209 \\ \hline
        adi & Alternating Direction Implicit solver & 98.1k & 55.2M & 717.1k & 431M & 42.1M & 25.3G & 0.54 & 298 \\ \hline
        atax & Matrix transpose and vector mult. & 9993 & 963.4k & 22.5k & 2.24M & 1.45M & 143M & 0.03 & 2.81 \\ \hline
        bicg & BiCG sub-kernel of BiCGStab solver & 6736 & 642.8k & 22.5k & 2.24M & 1.45M & 143M & 0.02 & 2.30 \\ \hline
        cholesky & Cholesky decomposition & 36.2k & 32.4M & 69.6k & 64.6M & 4.21M & 3.80G & 0.12 & 107 \\ \hline
        correlation & Correlation computation & 48.2k & 23.2M & 101.2k & 46.7M & 6.07M & 2.75G & 0.16 & 70.1 \\ \hline
        covariance & Covariance computation & 47.6k & 23.1M & 90.7k & 46.0M & 5.45M & 2.71G & 0.14 & 70.3 \\ \hline
        deriche & Edge detection filter & 133.5k & 11.1M & 401.5k & 33.9M & 24.8M & 2.09G & 0.35 & 26.8 \\ \hline
        doitgen & Multi-res. analysis kernel (MADNESS) & 42.8k & 22.6M & 70.4k & 43.3M & 4.22M & 2.55G & 0.12 & 67.4 \\ \hline
        durbin & Toeplitz system solver & 7509 & 22.6M & 10.1k & 965.2k & 604.5k & 56.7M & 0.02 & 1.68 \\ \hline
        fdtd-2d & 2-D Dinite Diff. Time Domain kernel & 108.7k & 43.2M & 377.5k & 157M & 22.3M & 9.27G & 0.45 & 178 \\ \hline
        floyd-warshall & Graph shortest path length & 662.7k & 17.6M & 2.81M & 75.8M & 164M & 4.44G & 3.12 & 83.7 \\ \hline
        gemm & Matrix-multiply $(C = \alpha A B + \beta C)$ & 49.1k & 32.0M & 139.7k & 95.4M & 8.34M & 5.61G & 0.16 & 104 \\ \hline
        gemver & Vector mult. and matrix add. & 15.0k & 1.45M & 55.2k & 5.4M & 3.39M & 332M & 0.05 & 4.86 \\ \hline
        gesummv & Scalar, vector and matrix mult. & 3839 & 251.8k & 14.8k & 1.00M & 1.01M & 68.2M & 0.01 & 1.07 \\ \hline
        gramschmidt & Gram-Schmidt decomposition & 59.2k & 34.9M & 115.3k & 69.7M & 6.93M & 4.10G & 0.18 & 106 \\ \hline
        heat-3d & Heat equation over 3D data domain & 73.1k & 34.1M & 925.2k & 494M & 54.4M & 28.9G & 0.84 & 413 \\ \hline
        jacobi-1d & 1-D Jacobi stencil computation & 3497 & 239.4k & 10.2k & 718.0k & 604.4k & 42.2M & 0.01 & 0.76 \\ \hline
        jacobi-2d & 2-D Jacobi stencil computation & 98.7k & 37.1M & 473.9k & 184M & 27.9M & 10.8G & 0.53 & 202 \\ \hline
        lu & LU decomposition & 69.1k & 64.5M & 131.8k & 128M & 7.86M & 7.55G & 0.24 & 212 \\ \hline
        ludcmp & LU decomposition + Forward Subst. & 76.6k & 1.83M & 141.7k & 3.58M & 8.45M & 211M & 0.22 & 5.02 \\ \hline
        mvt & Matrix Vector product and Transpose & 9945 & 963.2k & 22.6k & 2.24M & 1.46M & 143M & 0.03 & 2.74 \\ \hline
        nussinov & Dyn. programming for seq. alignment & 122.4k & 3.05M & 522.0k & 1.00M & 30.8M & 1.17G & 0.66 & 35.8 \\ \hline
        seidel-2d & 2-D seidel stencil computation & 89.8k & 47.7M & 785.8k & 428M & 46.2M & 25.1G & 0.83 & 451 \\ \hline
        symm & Symmetric matrix-mult. & 26.5k & 19.4M & 98.5k & 72.7M & 5.90M & 4.27G & 0.10 & 73.9 \\ \hline
        syr2k & Symmetric rank-2k update & 31.9k & 17.6M & 171.7k & 104M & 10.2M & 6.14G & 0.16 & 88.2 \\ \hline
        syrk & Symmetric rank-k update & 31.9k & 17.6M & 86.8k & 52.3M & 5.21M & 3.07G & 0.11 & 59.9 \\ \hline
        trisolv & Triangular solver & 2597 & 241.8k & 6637 & 642.4k & 520.2k & 50.1M & 0.01 & 0.79 \\ \hline
        trmm & Triangular matrix-mult. & 20.2k & 14.6M & 36.9k & 28.9M & 2.25M & 1.70G & 0.07 & 44.3 \\ 
    \end{tabular}
}
\end{table}

%% file: userfeedback.tex
\label{subsec:user-feedback}
We briefly outline some debugging features we have implemented, to assist developers in understanding \PEQC's errors. When using the verifier to assist the developer during their (semi-)manual program code transformations, we recommend testing with small problem sizes, since the execution time of the interpreter is roughly 0.5M statements interpreted per second. Then, once the program is valid on smaller problem sizes, longer verification can be afforded, on the final problem size. 

\paragraph*{Visualizing CDAG differences} A typical error manifests as a mismatch between CDAGs produced: the computations differ. We implement a simple visualizer, based on dotty, for the CDAGs of the mismatched memory location. It embeds the line number in the source program that produced every operation, and color-codes the first mismatched node. We illustrate below with a small excerpt of two trees produced for a stencil benchmark, raised to MLIR. While equivalence succeeds when using A/C normalization, it fails without: on the right the operation is (floatconstant * expr), on the left (expr * floatconstant). 
Note the tree printed here has hundreds of nodes, we only display its root.
\begin{figure}[h!tb]
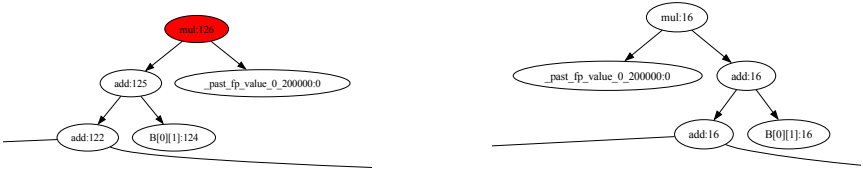

\vspace{-.5cm}
\includegraphics[width=0.35\textwidth]{tree1.pdf}
\hspace{1.5cm}\includegraphics[width=0.35\textwidth]{tree2.pdf}
\caption{Excerpt of CDAGs of two variants of the same 5-point stencil program, reordering the LHS and RHS of an operation. 
Line numbers help tracking the location of instructions which produced the operations, especially when temporary variables are used.}
\end{figure}

Optionally, we can also display the value of any variable marked by the user for debug at the time each operation (node in the cdag) was executed. We automatically compute the surrounding loop iterators and track them without user assistance, displaying the iteration at which an operation was executed to ease debugging.

\paragraph*{Reporting races and deadlocks}
Another major class of bug relates to incorrect parallelization and/or synchronization between programs. This is particularly critical when verifying optimizations for HLS that use blocking FIFOs and dataflow parallelization, such as the systolic array implementation of a matrix-multiply kernel \cite{fpga24}. 

In the case that the interpreter detects a race condition, the memory location (e.g. \texttt{temp[0][42]}) is displayed, as well as the line number and thread id of the instructions that accessed that location in a possible race. The source code of each thread involved in producing the race is then displayed. Optionally, the list of threads that executed and terminated before the race can be displayed.  

Our reporting for deadlocks, however, remains quite basic. We list all threads that are still active but cannot make further progress, along with their source code. We also report which semaphore/FIFO causes a deadlock (there may be more than one --- we only report the first one blocking further progress). Any violation of the rules for deterministic concurrency presented in Sec.~\ref{sec:formal} (e.g., multiple concurrent acquires) are reported also. Our debugging mode essentially displays the trace of all events involving semaphores (initialization, \texttt{wait}/\texttt{set}/\texttt{acquire}/\texttt{release}), whether they clear or are blocked, and whether a context-switch is needed. In our experience, this exhaustive trace enables the user to quickly locate culprits in case of synchronization bugs, but is best manipulated for small programs for easier debugging. Future work includes developing a recording of the synchronization operations executed by threads over the program lifetime, and displaying a forest of threads as nodes and their connections via semaphores/FIFOs instead of simply the trace.